\documentclass[journal,draftcls,onecolumn,12pt,twoside]{IEEEtran}
\normalsize
\IEEEoverridecommandlockouts

\usepackage{amsmath, amssymb, cite}
\usepackage{amsthm,bm,bbm}
\usepackage{mathtools}
\usepackage[small]{caption}
\usepackage{amsthm,multirow,color,amsfonts}
\usepackage{tabulary}
\usepackage{subfigure}
\usepackage{graphics,graphicx}
\usepackage{setspace}
\usepackage{enumerate}
\usepackage[]{algorithm2e}
\usepackage{comment}
\usepackage{hyperref}
\usepackage{subfigure}
		
\newtheorem{theo}{Theorem}

\newtheorem{prop}{Proposition}

\newtheorem{ex}{Example}

\def\compactify{\itemsep=0pt \topsep=0pt \partopsep=0pt \parsep=0pt}
\let\latexusecounter=\usecounter

\begin{document}
\title{Applications of Common Information to Computing Functions}
\author{Derya Malak
\thanks{The author is with the ECSE Dept RPI, Troy, NY 12180 USA (malakd@rpi.edu).}
}

\maketitle

\begin{abstract}

We design a low complexity distributed compression scheme for computing arbitrary functions of sources with discrete alphabets. We use a helper-based method that extends the definition of the G{\'a}cs-K{\"o}rner-Witsenhausen (GKW) common information to functional common information. The helper relaxes the combinatorial structure of GKW by partitioning the joint source distribution into nests imposed by the function, which ensures hierarchical cooperation between the sources for effectively distributed computing. By contrasting our approach's performance with existing efficient techniques, we demonstrate the rate savings in recovering function and source data.

Permutation invariant functions are prevalent in learning and combinatorial optimization fields and most recently applied to graph neural networks. We consider the efficient compression for computing permutation invariant functions in a network with two sources and one decoder. We use a bipartite graph representation, where two disjoint sets of vertices (parts) denote the individual source graphs and the edge weights capture the joint source distribution. 
We compress bipartite graphs by creating connected components determined by the function's distribution, accounting for the edge symmetries, and eliminating the low probability edges. We separately encode those edges and send them as refinements. 
Our approach can substitute high complexity joint decoding techniques and inform neural networks to reduce the computing time and reduce complexity.
\end{abstract}


\section{Introduction}
\label{intro}
We consider a function-oriented distributed source compression model in networks. The goal is efficient encoding of sources to effectively recover a functional representation which is an abstraction of a task. 
Representative tasks include data gathering \cite{mosk2006computing}, private-key encryption \cite{8849315}, coding for storage systems and modeling updates \cite{8292960,el2015synchronizing,shah2011distributed} in cloud computing, 
and multi-agent computing \cite{zhang2019multi} in edge computing applications.

In data compression, the average number of bits per symbol needed to encode a data source is characterized by its entropy rate. In distributed compression, data across multiple sources can be jointly compressed via exploiting the joint source distribution. In task specific applications, the rate of joint compression can be further reduced while ensuring zero error compression asymptotically. Here, the task can be represented using a function abstraction. The fundamental limits of compression can be determined by how sensitive the function is to the subsets of sources in dispersed / distributed environments.

In this paper, we consider a helper-based model for distributed encoding for functional compression via generalizing the G{\'a}cs-K{\"o}rner-Witsenhausen (GKW) common information. The helper approach relies on relaxing the combinatorial structure of the GKW common information to effectively exploit the structure imposed by the function on the joint source distribution. We call the generalized common information the {\emph {functional common information}}. The helper, which acts as a common encoder, provides hierarchical cooperation between the sources to effectively ensure the distributed computation of a function of the sources. To demonstrate the power of the functional common information model, we consider the {\emph {class of permutation invariant functions}}. 
Permutation invariant problems are important and hard, and have extensive applications in multi-agent reinforcement learning, meta-learning, supervised and unsupervised learning \cite{zaheer2017deep}, such as in estimation of population statistics \cite{poczos2013distribution}, in combinatorial problems and group theory \cite{Gesselsymmetry2016}, 
cosmology \cite{ntampaka2016dynamical}, and detection of symmetry in computer-aided design \cite{benso2008graph}, \cite{tate2000symmetry}.

\subsection{Related Work}
\label{related}
Distributed source compression has been widely explored to determine the minimum encoding rates to achieve a desired distortion in joint decoding. The fundamental bounds include the distributed source coding scheme of Slepian-Wolf \cite{SlepWolf1973}, the rate-distortion coding model of Wyner-Ziv with 
side information  \cite{wyner1976rate}, and its generalization to functional rate distortion \cite{doshi2007source}.

Practical implementation of the Slepian-Wolf scheme requires some topological constraints in the network, and this requires joint source-network approaches \cite{salamatian2016efficient}. To circumvent that various notions of common information exist in the  
literature, such as G{\'a}cs-K{\"o}rner-Witsenhausen (GKW) \cite{gacs1973common},  \cite{witsenhausen1975sequences}, Wyner \cite{wyner1975common}, Kumar-Li-Gamal \cite{KLG14}, as well as generalizations including approximate
and exact information-correlation functions \cite{yu2016generalized}. 
To the best of our knowledge, GKW is the right type of common information that enables low complexity source decompositions.
In \cite{salamatian2016efficient}, authors exploited GKW common information to design efficient multisource network problems with correlated sources. 
A helper that encodes GKW common information can also effectively learn the source coupling, e.g., minimum entropy coupling via a polynomial time approximation algorithm 
\cite{cicalese2019minimum}, \cite{li2020efficient}, or maximal coupling.

The family of permutation invariant functions has a special structure which enables to design a deep network architecture that can operate on sets and be deployed on a variety of scenarios including both unsupervised and supervised learning tasks \cite{zaheer2017deep}. 
%
%
Current neural networks (NNs) cannot exploit the equivalent states in networks. Recently, finite group NN architectures to create symmetry equivariant learning algorithms (for games such as go and chess) were devised in \cite{carroll2020finite} for saving computing time. 
Graph neural networks (GNNs) have been powerful in graph analysis and machine learning tasks, such as classification and regression \cite{zhou2018graph}. GNNs capture the dependence of graphs via message passing between the nodes \cite{zhou2018graph}. Unlike standard NNs, GNNs retain a state that can represent information from its neighborhood with arbitrary depth. 
In \cite{li2020pair} authors have designed an auto-encoder for graph data to efficiently encode and reconstruct the feature distribution for respective pairs and their surrounding context. 

In addition to NN architectures, techniques have been proposed to understand the limits of compression in networks, such as distributed compression for sources \cite{SlepWolf1973}, \cite{wyner1976rate}, computation across networks \cite{FM14,OR01,Kor73,AO96}, and graph data \cite{delgosha2020universal}. We are particularly interested in compressing bipartite networks which are significant in various natural systems ranging from information and economic systems to social networks, opinion networks, and recommendation systems \cite{zhou2007bipartite}.  
In \cite{zhou2007bipartite} authors have devised a weighted network model 
to compress bipartite networks by projecting on the set of vertices, i.e., the source nodes. This technique can perform better than global ranking method and collaborative filtering for personal recommendation systems. 

\subsection{Motivation and Contributions}
In task oriented settings, due to bottleneck links, privacy concerns, or compression costs, functional compression might be preferred over distributed encoding of data \cite{SlepWolf1973,wyner1976rate}. Hence, we consider the problem of distributed lossless compression for computing. Given jointly distributed source variables $X_1$ and $X_2$ with discrete alphabets, and a joint decoder our goal is to provide an error-free representation of an arbitrary function $f(X_1,\, X_2)$ at the receiver by characterizing the fundamental compression limits of $(X_1,\, X_2)$ for computing $f$.

To the best of our knowledge, there are no techniques that provide tight fundamental limits for compressing functions of distributed sources. This problem is challenging because of its distributed nature, and the existing approaches rely on graph coloring-based schemes that are NP-complete \cite{AO96,OR01,doshi2007source}. To that end, in this paper, we provide a different perspective to functional compression, which allows a partial extraction of the common information between $X_1$ and $X_2$. It also exploits the permutation invariance.

Our main contributions can be summarized as follows.

i. We provide a natural generalization of GKW common information to functional common information to relax the combinatorial structure of GKW. 
To compute a specific function, we determine how to extract the source features as well as the common randomness tailored to the specific task abstracted by the function. 

ii. We ensure an efficient zero error computation via a helper which exploits the functional coupling of the source random variables. The helper acts as a common encoder by creating nests which is the key into obtaining rate savings. The collection of nests is a proxy for the functional common information and ensures hierarchical cooperation between the sources to facilitate distributed compression for computing. 

iii. We demonstrate that the nested distributed compression scheme can exploit the functional common information and provide significant rate savings in functional compression 
as well as in recovering 
the sources themselves. We illustrate the savings via several scenarios which we detail in Sect. \ref{FunctionalCommonEntropy}.

iv. 
We characterize the fundamental rate limits for distributed and zero error compression for computing permutation invariant functions with two data sources and one destination. In other words, the function outcome does not depend on the spatial locations or permutations of the data. Due to the distributed nature of the problem, where the sources cannot communicate, it is easier to compute the function via extracting the common randomness. To that end, we employ a helper, possibly with a limited rate, to partially exploit the joint source distribution, 
the symmetry incurred by the permutation invariance, as well as the low probability events. Extracting the common information eliminates the complexity of the joint decoding in the receiver.



{\bf Notation:} Capital letter $X$ denotes a discrete random variable, and small letter $x$ its realization. 
$P_{X_1}$ and $P_{X_1,\, X_2}$ denote the distribution of $X_1$ and joint distribution of $X_1$ and $X_2$, respectively. We denote the probability of an event $A$ by $P(A)$. We use $H(X)$ to represent the entropy of $X$, and denote the entropy function for distribution $P_{X}=(p_1,\,p_2,\,\hdots)$ by $h(P_X)=-\sum\limits_{i} p_i\log p_i$ where the logarithm is in base $2$. The notation $G_{X_1}$ denotes the characteristic graph of $X_1$ for computing $f(X_1,\, X_2)$ and $c_{G_{X_1}}$ is its valid (vertex) coloring. We use $H_{G_{X_1}}(X_1)$ to represent the graph entropy of $X_1$ for computing $f(X_1,\, X_2)$. Similarly, $H_{G_{X_1}}(X_1\vert X_2)$ represents the conditional graph entropy of $X_1$ for computing $f(X_1,\, X_2)$ given $X_2$. We provide the detailed definitions in Sect. \ref{background}.


\section{Technical Background}
\label{background}
We next provide a baseline for our analysis, which covers the related work on distributed compression for recovering source variables or functions of sources at the destination.

\subsection{Distributed Source Compression}
\label{SW}
Consider the joint decoding of two independently encoded source sequences with alphabets $\mathcal{X}_1$ and $\mathcal{X}_2$  modeled by random variables $X_1$ and $X_2$, which are jointly distributed according to $P_{X_1,\, X_2}$. The Slepian-Wolf theorem, by exploiting joint typicality, gives a theoretical lower bound for the lossless coding rate for distributed coding of the two statistically dependent i.i.d. finite alphabet source sequences as \cite{SlepWolf1973}:
\begin{align}
\label{rateregionSW}
R_1 \geq H(X_1\vert X_2),&\quad
R_2 \geq H(X_2\vert X_1),\nonumber\\
R_1+R_2 &\geq H(X_1,\, X_2).
\end{align}
Hence, it is necessary and sufficient to separately encode $X_1$ and $X_2$ at rates $(R_1, R_2)$ satisfying (\ref{rateregionSW}) to jointly recover $(X_1,\, X_2)$ at a receiver with zero error probability asymptotically.

\subsection{Distributed Functional Compression}
\label{GraphCompression}

The rate region of the distributed compression of the function $f(X)$ on the source data $X=\{X_1,\ldots,X_L\}$ depends on the function and the mappings from the source variables $X_l$, $l=1,\,2,\,\hdots, \,L$, to the destinations which challenge the codebook design due to the correlations among $X$ and $f$. To that end, the characteristic graph and its coloring play a critical role in the study of the fundamental limits of functional compression. 
A coloring of a graph $G=(V,E)$ is an assignment such that no two adjacent vertices have the same color. An $m-$coloring of a graph $G$ is a coloring with $m$ colors. 
Each vertex of the characteristic graph represents a possible different sample value, and two vertices are connected if they should be distinguished. More precisely, for the set of sources $X$ each taking values in the same alphabet $\mathcal{X}$, and function $f$, to construct the characteristic graph of $f$ on $X_1$, i.e., $G_{X_1}$, we draw an edge between vertices $u$ and $v \in \mathcal{X}$, if $f(u, x_2,\ldots,x_L) \neq f(v, x_2,\ldots, x_L)$ for any $\{x_2,\,\ldots,\,x_L\}$ whose joint instance has a non-zero measure. Surjectivity of $f$ determines the connectivity and coloring of the graph. The entropy rate of the coloring of $G_{X_1}$ characterizes the minimal representation needed from $X_1$ to reconstruct with fidelity the desired function $f$ \cite{Kor73}. The degenerate case of the identity function corresponds to having a complete $G_{X_1}$.

Given two jointly distributed source according to $P_{X_1,\, X_2}$, a graph $G_{X_1} = (V_{X_1} , E_{X_1} )$ of $X_1$ with respect to $X_2$, 
and function $f(X_1, X_2)$,
K{\"o}rner's graph entropy is expressed as \cite{Kor73,AO96}
\begin{align}
H_{G_{X_1}}(X_1)= \min\limits_{X_1\in W_1\in S(G_{X_1})} I(X_1; W_1), 
\end{align}
where $S(G_{X_1})$ is the set of all maximal independent sets of $G_{X_1}$, where a maximal independent set is not a subset of any other independent set, which is a set of vertices in which no two vertices are adjacent \cite{moon1965cliques}, and $X_1 \in W_1 \in S(G_{X_1})$ represents the distributions $p(w_1, x_1)$ such that $p(w_1,x_1) > 0$ implies $x_1 \in w_1$, where $w_1$ is a maximal independent set of $G_{x_1}$.

In a distributed setting, similar to the rate region in (\ref{rateregionSW}) that jointly recovers $(X_1,\, X_2)$, the rate region for effective zero error computation of $f(X_1,\, X_2)$ with arbitrary discrete inputs $X_1$ and $X_2$ has been derived in \cite{FM14,OR01}:
\begin{align}
\label{rateregiongraph}
R_1\geq H_{G_{X_1}}(X_1\vert X_2),&\quad
R_2\geq H_{G_{X_2}}(X_2\vert X_1), \nonumber\\
R_1+R_2\geq &H_{G_{X_1},G_{X_2}}(X_1, X_2),
\end{align}
where $G_{X_l}$ denotes the characteristic graph that source $l=1,\,2$ builds to distinguish the source outcomes that yield a different output for any value of $X_2$, and compute $f(X_1,\, X_2)$. Hence, the characteristic graphs built by each source are correlated. In (\ref{rateregiongraph}), $H_{G_{X_1}}(X_1)$ is K{\"o}rner's graph entropy, and $H_{G_{X_1},G_{X_2}}(X_1,\, X_2)$ is the joint graph entropy for functional compression \cite{FM14}. 

In functional compression, the Slepian-Wolf compression rate region in (\ref{rateregionSW}) provides an inner bound versus K{\"o}rner's graph entropy in (\ref{rateregiongraph}) that provides a convex outer bound and determines the limits of the functional compression. 
The region between the Slepian-Wolf compression bound and K{\"o}rner's graph entropy bounds indicates that there could be potentially a lot of benefit in exploiting the function's compressibility to reduce communication. 
%
For a comprehensive review of graph entropy and related definitions, see \cite[Ch. 21.1]{el2011network}.

\begin{figure}[t!]
    \centering
    \includegraphics[width=\columnwidth]{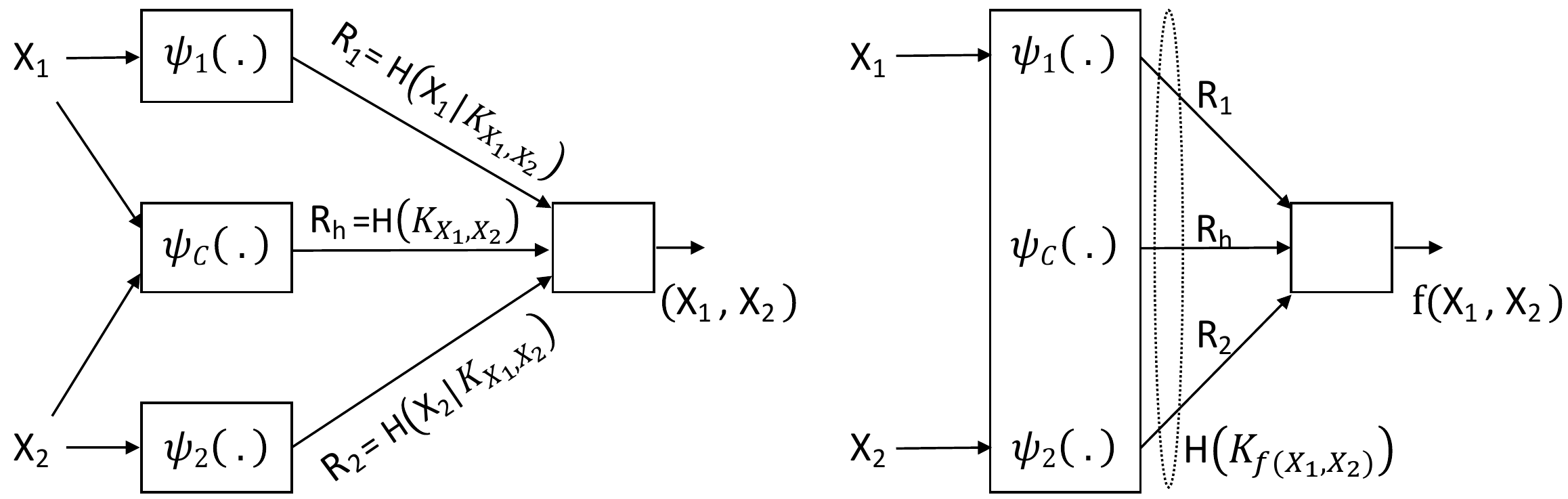}
    \caption{Source compression with helper. (Left) Distributed source compression with GKW common information with encoders private encoders $\psi_1$ and $\psi_2$ and a common encoder (helper) $\psi_C$. (Right) Distributed compression with  functional common information where the helper $\psi_C$ has access to partial information. In each setting, $R_h$ denotes the helper rate.} 
    \label{fig:FunctionalCompressionwithHelper}
\end{figure}

In general, finding minimum entropy colorings of characteristic graphs is NP-hard, and the optimal rate region of functional compression is an open problem \cite{el2011network}. However, in some instances, e.g., in \cite{DSME10} and \cite{FM14}, it is possible to compute these colorings efficiently. In \cite{FM14}, authors have shown that sending colorings of sufficiently large power graphs $G^n_{X_1}$ which satisfy coloring connectivity condition followed by Slepian-Wolf compression leads to an achievable coding scheme. 
The receiver performs a minimum entropy decoding \cite{csiszar2011information} on its received information and achieves $c_{G^n_{X_1}}$, and uses a look-up table to compute its function \cite{FM14}. The entropy of the characteristic graph $G_{X_1}$ is \begin{align}
H_{G_{X_1}}(X_1)=\lim_{n\to\infty} \min\frac{1}{n} H(c_{G_{X_1}^n}(X_1)), \nonumber
\end{align}
where $G_X^n$ is the n-th power of a graph $G_X$ and $c_{G_{X}^n}$ is a valid coloring of $G_X^n$. We next consider an example to show the savings of a coloring-based approach. Next, via Example \ref{uniform_example}, we numerically demonstrate that assigning colors to the power graphs of $G_{X_1}$ can compress the source more.

\begin{ex}
\label{uniform_example}
Consider the permutation invariant function $f(X_1,\, X_2)=|X_1|+|X_2|$ where the source variables $X_1$ and $X_2$ are uniformly distributed over $\mathcal{X}=\{-2,\,-1,\,0,\,1,\,2\}$. Hence, $G_{X_1}$ and $G_{X_2}$ are identical due to symmetry. The function takes different values if $X_1$ belongs to the sets $\{-2,2\},\{-1,1\},\{0\}$ which are equivalence classes for $X_1$. The values of $X_1$ in each class do not need to be distinguished because they yield the same function value independent of the value of $X_2$. However, different classes need to be distinguished. 
Hence, $G_{X_1}$ has edges across the equivalence classes and no edges within any class. The same logic holds for $X_2$. Therefore, we only need $3$ different colors 
instead of $5$ to distinguish the possible outcomes of the function. In this example, $P_{c_{G_{X_1}}}=\left(\frac{2}{5},\,\frac{2}{5},\,\frac{1}{5}\right)$ and the colorings satisfy $H(c_{G_{X_1}})=1.52$. The coloring is refined if we consider the power graphs of $X_1$, e.g., 8 colors will suffice to distinguish the distinct outcomes in the second power graph and $\frac{1}{2}H(c_{G^2_{X_1}})=1.48$ \cite{FM14}. In the limit $\frac{1}{n}H(c_{G^n_{X_1}})$ converges to $H_{G_{X}}(X)$. 
In case of no compression, $H(X_1)=\log 5=2.32$ and $H_{G_{X_1}}(X_1)=1.52$. Hence, the saving of functional compression over data compression is $0.8$ bits per source.
\end{ex}

\subsection{G{\'a}cs-K{\"o}rner-Witsenhausen Common Information}
\label{GacsKorner_CI}
The GKW common information plays a significant role in distributed source compression. When joint information extraction is not possible, it represents the maximum amount common information that can be separately extracted from individual sources $X_1$ and $X_2$. As shown in Fig. \ref{fig:FunctionalCompressionwithHelper} (Left), the GKW scheme consists of a helper $\psi_C$ that encodes the common information, as well as the separate encoders $\psi_1$ and $\psi_2$. The joint distribution $P_{X_1,\, X_2}$ has a bipartite graph representation $\mathcal{B}_{X_1,\, X_2}$. We illustrate an example $P_{X_1,\, X_2}$ along with its bipartite graph representation $\mathcal{B}_{X_1,\, X_2}$ in Fig. \ref{fig:GK_CI}. In this example, the empty entries in $P_{X_1,\, X_2}$ denote the coordinates where the joint probability mass equals zero and the realizations of $X_1$ and $X_2$ are indicated on the graph vertices.

The bipartite graph $\mathcal{B}_{X_1,\, X_2}$ has the following properties \cite{salamatian2016efficient}: (i) There is an edge between nodes $n_{x_1}$ and $n_{x_2}$ where $x_1\in \mathcal{X}_1,\, x_2 \in \mathcal{X}_2$, if and only if $P_{X_1}(X_1=x_1)>0$ and $P_{X_2\vert X_1}(X_2 =x_2\vert X_1 =x_1)>0$.
(ii) Each connected component $\mathcal{C}$ of $\mathcal{B}_{X_1,\, X_2}$ has a probability $$p(\mathcal{C}) = \sum\limits_{n_{x_1},n_{x_2}\in \mathcal{C}} P_{X_1,\, X_2}(X_1=x_1,\,X_2=x_2).$$

The common information decomposition of $P_{X_1,\, X_2}$ partitions $\mathcal{B}_{X_1,\, X_2}$ into a set $\mathcal{K}$ of a maximal number of connected components $\mathcal{C}_1,\dots,\mathcal{C}_{|\mathcal{K}|}$ where $|\mathcal{K}|$ is their cardinality. The GKW common information $K_{X_1,\, X_2}$ represents the index of the connected component in $\mathcal{B}_{X_1,\, X_2}$, and has distribution $P_{\mathcal{C}}=(p(\mathcal{C}_1),\dots, p(\mathcal{C}_{|\mathcal{K}|}))$. Hence, the entropy of $K_{X_1,\, X_2}$ equals
\begin{align}
\label{GK_CI}
H(K_{X_1,\, X_2} ) = H(\mathcal{C}) = \sum\limits_{k\in \mathcal{K}} p(\mathcal{C}_k)\log\left(\frac{1}{p(\mathcal{C}_k)}\right).     
\end{align}

Note that (\ref{GK_CI}) is the GKW common information between $X_1$ and $X_2$  \cite{gacs1973common}, where the common information variable $K_{X_1,\, X_2}$ is given as
\begin{align}
K_{X_1,\, X_2}=  \underset{H(Y\vert X_1)=H(Y\vert X_2 )=0}{\arg\max}H(Y). 
\end{align}

For the bipartite graph representation in Fig. \ref{fig:GK_CI}, the GKW common information $K_{X_1,\, X_2}$ distinguishes the two bipartitions with probabilities $\frac{1}{3}$ and $\frac{2}{3}$. Hence, $H(K_{X_1,\, X_2} )= h(\frac{1}{3})=0.918$.

As a result of data processing it holds that $H(X_1)\geq H(X_1\vert K_{X_1,\, X_2})\geq H(X_1\vert X_2)$, and similarly $H(X_2)\geq H(X_2\vert K_{X_1,\, X_2})\geq H(X_2\vert X_1)$. The implication of this result is that through a helper, it is possible to effectively compress sources in a distributed setting. We next apply GKW common information to distributed functional compression of data.

\section{GKW Common Information for Functional Compression}
\label{GKW_Functional}

We next provide a helper-based scheme to independently encode the distributed sources towards effectively computing the function $f(X_1,\, X_2)$ at the decoder by incorporating GKW common information $K_{X_1,\, X_2}$ between $X_1$ and $X_2$. 
\begin{prop}\label{Coding_GK_f(X_1_X_2)} {\bf Coding with helper for functional compression.} 
There exists a helper-based encoding scheme for efficient zero error compression of $(X_1,\, X_2)$ to asymptotically compute $f(X_1,\, X_2)$ that operates at rates:
\begin{align}
\label{GK_rate_region_function}
R_l&\geq H_{G_{X_l}}(X_l\vert K_{X_1,\, X_2}), \quad l=1,\,2,\nonumber\\
R_1+R_2&\geq H(K_{X_1,\, X_2})+\sum\limits_{l=1,\,2}H_{G_{X_l}}(X_l\vert K_{X_1,\, X_2}),
\end{align}
where $G_{X_i}$ denotes the characteristic graph that source $i$ determines for computing $f(X_1,\, X_2)$ given $K_{X_1,\, X_2}$.
\end{prop}

The advantages of the helper scheme over Slepian-Wolf-based compression are twofold. The helper not only eliminates the exponential complexity of joint typicality decoding but also provides an asymptotically zero-error guarantee. On the other hand, the rate region in Prop. \ref{Coding_GK_f(X_1_X_2)} is encompassed by that of the distributed functional compression scheme in \cite{FM14}. When $K_{X_1,\, X_2}=0$, the result in (\ref{GK_rate_region_function}) boils down to independent compression of the characteristic graphs of the sources because even though sources may be dependent, they do not have a joint distribution which is the block-diagonal form. While we skip the proof due to space constraints, we numerically contrast the two schemes later in Sect. \ref{numerical} in Fig. \ref{fig:FunctionalCommonInformationRate}. We also note that Prop. \ref{Coding_GK_f(X_1_X_2)} may require small marginal rates when $H(K_{X_1,\, X_2})$ is high, and may provide a similar rate region as in \cite{FM14}. Achieving low marginal rates is possible because having a $P_{X_1,\, X_2}$ with a high number of connected components 
ensures independent trimming of the source codebooks across $\{\mathcal{C}_i\}_{i=1}^k$. 

The special case of Prop. \ref{Coding_GK_f(X_1_X_2)} is the rate region derived for efficient zero error encoding and decoding of the identity function $(X_1,\, X_2)$ in \cite{salamatian2016maximum}. In this case, it is clear that the graph entropy $H_{G_{X_i}}(X_i)$ equals the Shannon entropy $H(X_i)$ because each vertex forms an independent set in $G_{X_i}$, and needs to be distinguished to determine each unique pair $(X_1,\, X_2)$. 
%

\begin{figure*}[t!]
    \centering
    \includegraphics[width=0.7\textwidth]{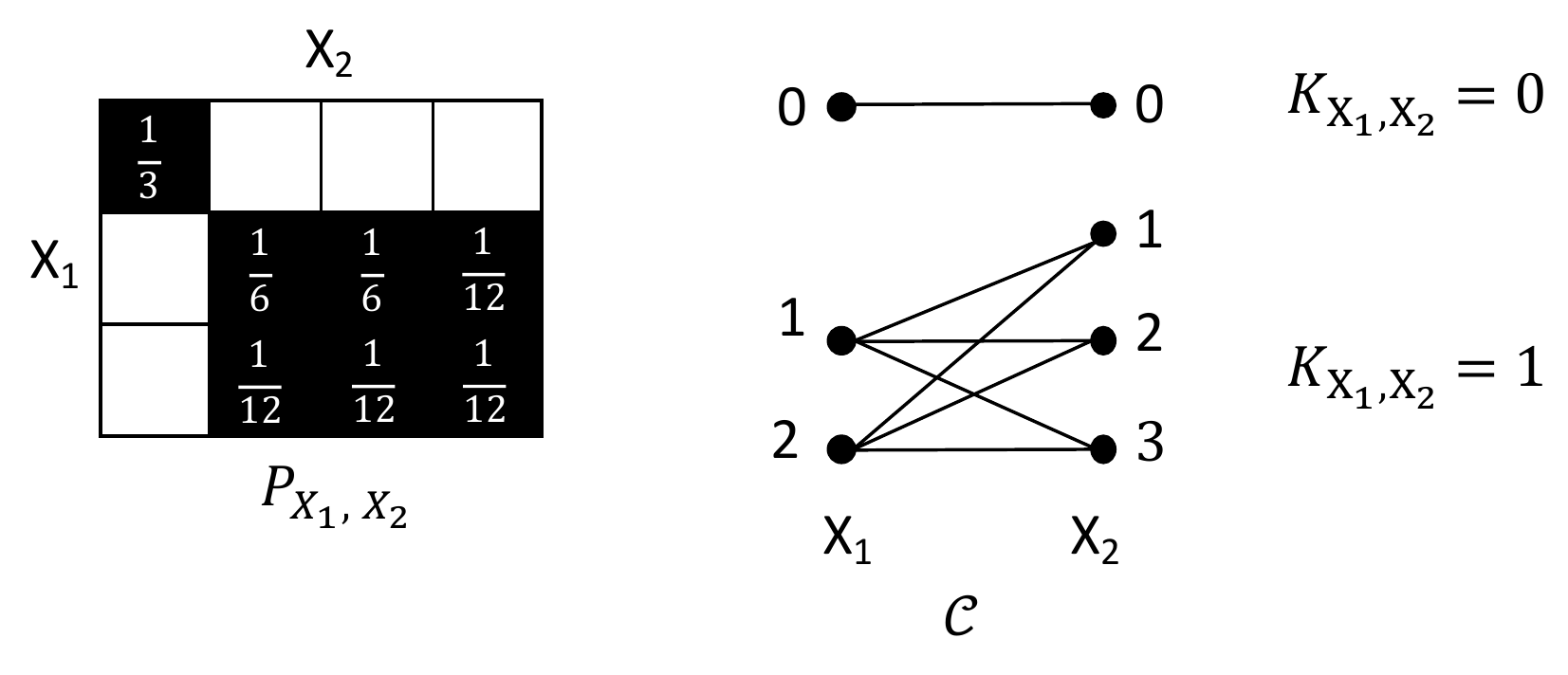}
    \caption{Joint probability matrix of $X_1$ and $X_2$, $P_{X_1,\, X_2}$, and GKW common information graph, $\mathcal{C}$ \cite{salamatian2016efficient}.}
    \label{fig:GK_CI}
\end{figure*}

\begin{figure*}[t!]
    \centering
    \includegraphics[width=\textwidth]{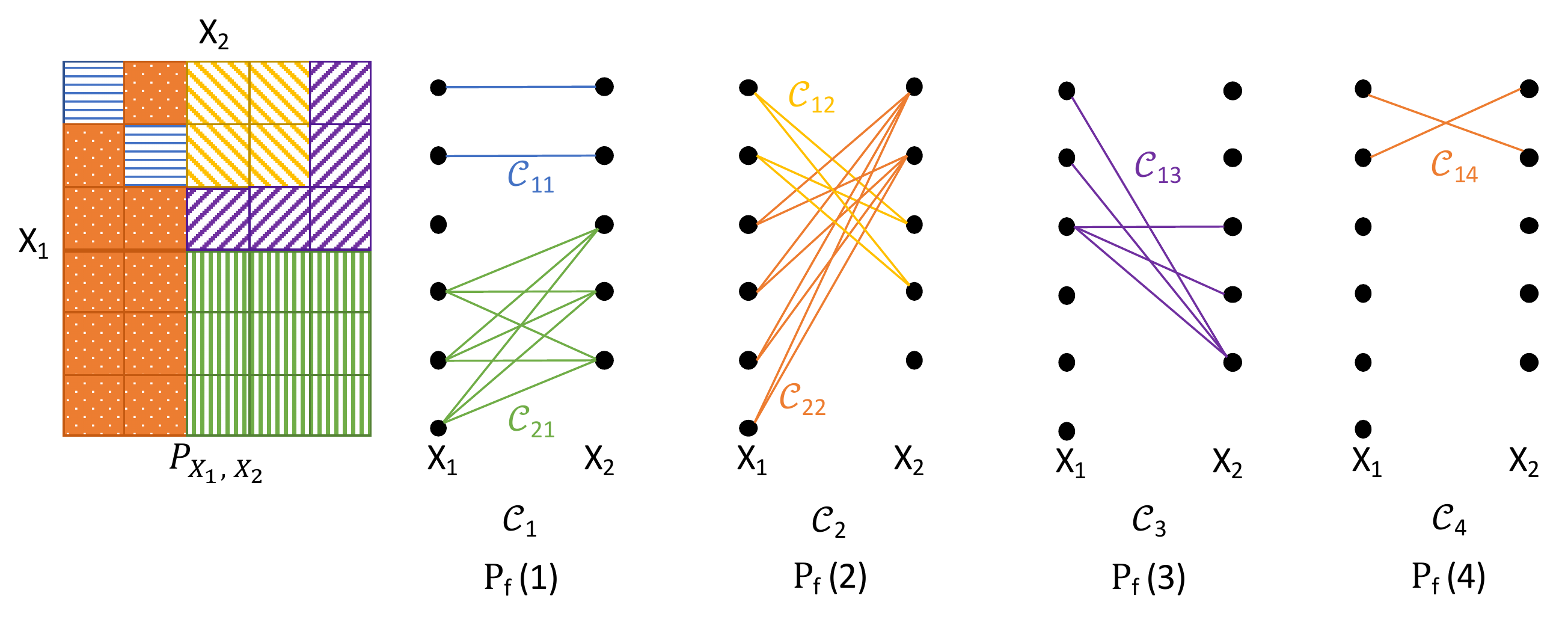}
    \caption{$P_{X_1,\, X_2}$ and functional common information graph where different colors represent distinct function outcomes (5 outcomes in total) and graph $i\in\mathcal{I}$ corresponds to a distinct outcome of  the nested random variable $V=i$. Here, given $V=i$, where $i\in\{1,\,2,\,3,\,4\}$, the resultant graph $\mathcal{C}_i$ has connected components $\mathcal{C}_{ki}$ for $k=1,\,2,\,\dots$. Given $V$ only one of the sources needs to transmit to specify the connected component, i.e., determine the function outcome.}
    \label{fig:Functional_CI}
\end{figure*}

Building on the concepts in this section, we next introduce a new measure which we call functional common entropy.

\section{Functional Common Information}
\label{FunctionalCommonEntropy}
Classical data oriented compression techniques are not tailored to the task specific or function-oriented aspects. In this section we introduce a common information-based compression model for computing functions. We extend the bipartite graph representation of GKW common information in \cite{salamatian2016efficient} to functional common information which is a generalized nested description for identifying a connected component $\mathcal{C}$ in $\mathcal{B}_{X_1,\, X_2}$. This measure is jointly determined by the source coupling $P_{X_1,\, X_2}$ as well as by $f(X_1,\, X_2)$. 

Let $V$ be a nested random variable over the set $\mathcal{I}$ such that each $V=i$, $i\in \mathcal{I}$ embeds different sets of connected components. Specifically, conditioning on the nest index $V=i$ we create a graph $\mathcal{C}_i$ such that the collection of $\{\mathcal{C}_i\}$ satisfies $\mathcal{C}=\bigcup_{{i\in\mathcal{I}}\,} \mathcal{C}_i$ which is a union of vertices and edges in $i\in\mathcal{I}$. 
Each $\mathcal{C}_i$ has a set of connected components $\{\mathcal{C}_{ki}\}_{k\in \mathcal{K}_i}$ where $|\mathcal{K}_i|$ is their count. The common information of $\mathcal{C}_i$ represents the index of the connected component $\{\mathcal{C}_{ki}\}_{k\in \mathcal{K}_i}$. 
Nest $i$ ensures a maximal number of connected components $|\mathcal{K}_i|$ such that the subgraph $\mathcal{C}_{ki}$, $k\in \mathcal{K}_i$ represents a distinct function outcome $f(X_1,\, X_2)$ that is jointly encoded conditional on $V=i$. 
The index of the connected component in $\mathcal{C}_i$ has distribution $(p(\mathcal{C}_{1i}),\hdots p(\mathcal{C}_{|\mathcal{K}_i|i}))$ and its entropy is $H(\mathcal{C}_i)$. 
We define the entropy of the functional common information as
\begin{align}
\label{GK_functional_CI}
H(K_{f(X_1,\, X_2)} ) &= H(V,\mathcal{C}_1,\hdots,\mathcal{C}_I)\\
&=H(V)+\sum\limits_{i\in\mathcal{I}} P_f(i) H(\mathcal{C}_i \vert V=i)\\
&=H(V)+ \sum\limits_{i\in\mathcal{I}}P_f(i)\sum\limits_{k\in\mathcal{K}_i} p(\mathcal{C}_{ki})\log\left(\frac{1}{p(\mathcal{C}_{ki})}\right),    
\end{align} 
where $K_{f(X_1,\, X_2)}$  is the functional common information between $X_1$ and $X_2$ to compute $f(X_1,\, X_2)$, and is given as
\begin{align}
\label{GK_functional_CI_argument}
K_{f(X_1,\, X_2)}=  \underset{U_i\to V\to U_j,\,\, i\neq j,\,\forall i,\,j\in\mathcal{I}}{\underset{H(U_i\vert V,\,f(X_1,\, X_2))=0}{\arg\max}} H(V,\,U_1,\,\hdots,\,U_I).
\end{align}

The nested distributed function encoding scheme requires a helper, i.e., common source encoder, for extracting $V$. The helper $\psi_C$ determines the nests $\mathcal{I}$ based on $P_{X_1,\, X_2}$ and $f(X_1,\, X_2)$. 
Once $\psi_C$ extracts the functional coupling and sends $V=i$ to the decoder via a separate link, this scheme resembles GKW. The distributed source encoders $\psi_1$ and $\psi_2$ can then build functional common information graphs $\{\mathcal{C}_{ki}\}_{k\in \mathcal{K}_i}$ and send them to the decoder at sufficient rates to determine the function outcomes. 
In other words, $\psi_C$ provides an a priori coupling between the sources to generate particular function outcomes. This ensures a high level hierarchical cooperation between $\psi_1$ and $\psi_2$ to distinguish $\{\mathcal{C}_{ki}\}_{k\in \mathcal{K}_i}$ 
in each nest $V$. 
In this scheme, sources can only partially cooperate through the helper. Hence, extraction of the functional common information does not guarantee sending of the information at a rate $H(f(X_1,\, X_2))$. 
We illustrate the functional common information scheme that encompasses GKW in Fig. \ref{fig:FunctionalCompressionwithHelper} (Right).

We illustrate the measure of functional common information via an example in Fig. \ref{fig:Functional_CI}. Here, $P_{X_1,\, X_2}$ is joint probability table where row $l=0,\dots,5$ corresponds to $X_1=l$, and column $l=0,\dots,4$ corresponds to $X_2=l$. In $P_{X_1,\, X_2}$ each distinct function outcome $f(X_1,\, X_2)$ is represented by a different pattern. The variable $K_{f(X_1,\, X_2)}$ is jointly determined by $V$ that takes values $i=1,\,2,\,3,\,4$ and divides $P_{X_1,\, X_2}$ into distinct graphs $\mathcal{C}_i$ such that each resultant graph $\mathcal{C}_i$ has connected components $\mathcal{C}_{ki}$ for $k=1,\,2,\hdots, |\mathcal{K}_i|$ where $|\mathcal{K}_i|$ is the cardinality of $\mathcal{K}_i$. The corresponding nested bipartite graph representation is indicated along with $P_{X_1,\, X_2}$. 

\begin{figure*}[t!]
    \centering
    \includegraphics[width=\textwidth]{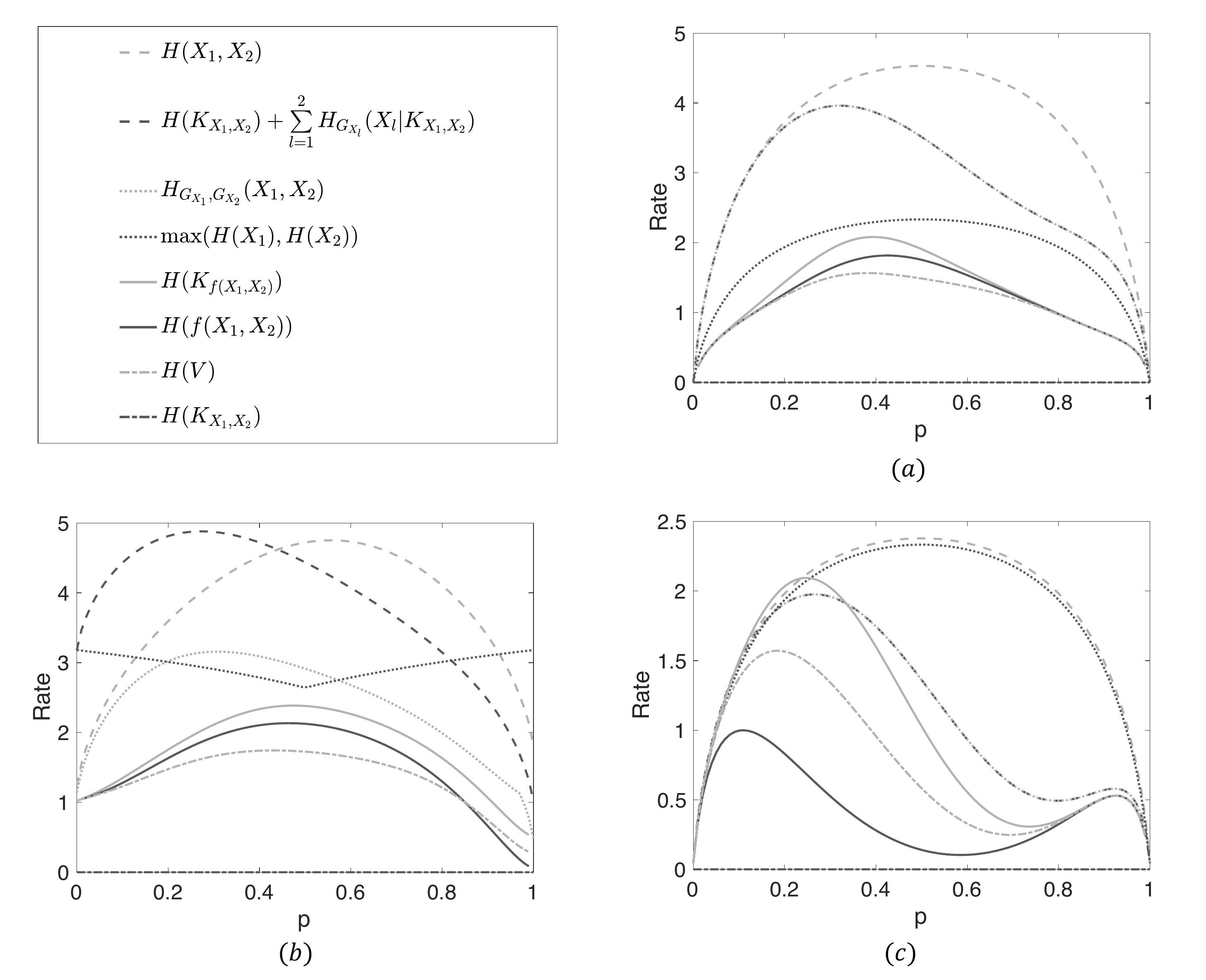}
    \caption{(a) Rates for functional compression in Example Fig. \ref{fig:Functional_CI} where $X_1$ and $X_2$ are independent binomially distributed random variables with success probabilities $p$ and $1-p$, respectively. (b) Rates for functional compression where $X_1$ and $X_2$ are jointly distributed 
    random variables. (c) Rates for functional compression where $X_1$ and $X_2$ are independent, binomially distributed, and with success probabilities $p$ and $0.001$, respectively. The common legend for three figures is shown in top-left.}
    \label{fig:FunctionalCommonInformationRate}
\end{figure*}

\subsection{Comparison of Functional Common Information with GKW Common Information}
The functional common information $K_{f(X_1,\, X_2)}$ of the mixture $\{\mathcal{C}_i\}_{i\in\mathcal{I}}$ ensures a higher entropy than that of $K_{X_1,\, X_2}$ as each $\mathcal{C}_i$ has a higher number of connected components than that of $\mathcal{C}$.  
Alternatively, we obtain $K_{X_1,\, X_2}$ through aggregation of $\{\mathcal{C}_i\}_{i\in\mathcal{I}}$ which reduces the entropy: 
\begin{align}
H(K_{X_1,\, X_2})&=H(\mathcal{C})=H\Big(\bigcup_{{i\in\mathcal{I}}\,} \mathcal{C}_i\Big) \nonumber\\
&\leq H(V)+H(\mathcal{C}_V \vert V)\nonumber\\
&=H(K_{f(X_1,\, X_2)}).\nonumber
\end{align}
Hence, if the helper can extract information about $f(X_1,\, X_2)$ via partially exploiting the joint structure in $P_{X_1,\, X_2}$, then the amount of common information extracted is larger.

A zero-error distributed code for decoding the pair $(X_1,\, X_2)$ using the functional common information is developed below.

\begin{prop}\label{zero_error_X1X2_functional_CI} {\bf Coding with helper for joint compression.}
There exist an efficient zero error encoding and decoding of $X_1$ and $X_2$ that operates at rates:
\begin{align}
\label{helper_based_encoding}
R_1&\geq H(X_1\vert K_{f(X_1,\, X_2)}), \nonumber\\ R_2&\geq H(X_2\vert K_{f(X_1,\, X_2)}),\nonumber\\
R_1+R_2&\geq H(K_{f(X_1,\, X_2)})+\sum\limits_{l=1,\,2}H(X_l\vert K_{f(X_1,\, X_2)}),
\end{align}
where $H(X_1\vert K_{f(X_1,\, X_2)})$ and $H(X_2\vert K_{f(X_1,\, X_2)})$ are the additional rates required from respective sources given $K_{f(X_1,\, X_2)}$.
\end{prop}

Prop. \ref{zero_error_X1X2_functional_CI} gives the limits for distributed extraction $(X_1,\, X_2)$ in the presence of $K_{f(X_1,\, X_2)}$. 
Because $K_{f(X_1,\, X_2)}=g(X_1,\, X_2)$ for some function $g$ and $H(X_1\vert K_{f(X_1,\, X_2)})$ is an aggregation of $X_1$, $H(X_1)\geq H(X_1\vert K_{f(X_1,\, X_2)})\geq H(X_1\vert K_{f(X_1,\, X_2)},X_2)\geq 0$. If $g$ is partially invertible the last inequality becomes an equality. 
In general, unlike the GKW common information $K_{X_1,\, X_2}$ it may not be true that $H(X_1\vert K_{f(X_1,\, X_2)})\geq H(X_1\vert X_2)$. For example, if $X_1$ and $X_2$ are independent and $g(X_1,\, X_2)$ is the binary xor function with $X_1\in\mathcal{X}_1=\{0,\,1,\,2,\,3\}$ and $X_2\in\mathcal{X}_2=\{0,\,1\}$, then it holds that $H(X_1\vert X_2)\geq H(X_1\vert K_{f(X_1,\, X_2)})$. This implies for some functions $f(X_1,\, X_2)$ that the helper provides significant savings in terms of the marginal rates at which each source operates, and ensures efficient decoding of $(X_1,\, X_2)$ via $K_{f(X_1,\, X_2)}$.

We next extend Prop. \ref{zero_error_X1X2_functional_CI} to provide a zero-error encoding and decoding of $f(X_1,\, X_2)$. We show that the nested scheme only requires either of the sources to send information.

\begin{prop}\label{functional_common_information_two_sources}
There exist an efficient zero error encoding and decoding of $f(X_1,\, X_2)$ that operates at rates:
\begin{align}
R_1&\geq H_{G_{X_1}}(X_1\vert V)=\sum\limits_{i\in\mathcal{I}} P_f(i) H_{G^i_{X_1}}(X_1\vert V=i),\nonumber\\
R_2&\geq H(V),
\end{align}
where $H_{G^i_{X_1}}(X_1\vert V=i)$ denotes the rate of encoding required from $X_1$ to compute $f(X_1,\, X_2)$ when the helper specifies the resultant graph $\mathcal{C}_i$ which has a probability $P_f(i)$ for $i\in\mathcal{I}$. 
\end{prop}

From Prop. \ref{functional_common_information_two_sources}, note it suffices that given nest index $V=i$ only one source needs to encode and transmit to identify the right connected component $\mathcal{C}_{ki}$, $k\in \mathcal{K}_i$. Each $\mathcal{C}_{ki}$ denotes a particular function outcome specified by a color in Fig. \ref{fig:Functional_CI}.

Functional compression is more effective versus encoding data itself, i.e., $H_{G^i_{X_1}}(X_1\vert V=i)\leq H(X_1\vert V=i)$ for $i\in\mathcal{I}$. From (\ref{GK_functional_CI_argument}) it is also true for $l=1,\,2$ that $H_{G_{X_l}}(X_l\vert V)\geq H_{G_{X_l}}(X_l\vert K_{f(X_1,\, X_2)})=H_{G_{X_l}}(X_l\vert V,\mathcal{C}_1,\hdots,\mathcal{C}_I)$. 
The rate savings in functional compression are more noticeable if trimming of the sources is viable via fewer partitions, i.e., $|\mathcal{I}|$ and $H(V)$ are small. This is because mixing of $\{\mathcal{C}_i\}_{i\in\mathcal{I}}$ increases $H(K_{f(X_1,\, X_2)})$, i.e., the rate of common encoding. 
The benefit of operating at rates in Prop. \ref{functional_common_information_two_sources} is eminent when both $H(V)$ and $H_{G_{X_1}}(X_1\vert V)$ are small. 
If $H(V)$ is small, $P_{X_1,\, X_2}$ can pave the way for computing the function, because independent trimming of the source set is feasible. For example, see the orange blocks with dots and the green blocks with vertical lines in Fig. \ref{fig:Functional_CI}. 
However, $H_{G_{X_1}}(X_1\vert V)$ can still be as high as $H_{G_{X_1}}(X_1)$. On the other hand, if $H(V)$ is large, there may be no benefit of encoding at the rates given in Prop. \ref{functional_common_information_two_sources}. For example, when $f(X_1,\, X_2)$ is the identity function such that each source pair needs to be distinguished, then $H(V)\geq H_{G_{X_1,\, X_2}}(X_1,\, X_2)$ and $H_{G_{X_1}}(X_1\vert V)=0$, where $H_{G_{X_1,\, X_2}}(X_1,\, X_2)$ is the joint encoding rate for distributed functional compression without a helper characterized in \cite{FM14}.

While it is nontrivial to quantify the fundamental limits in Prop. \ref{functional_common_information_two_sources}, we next run some numerical experiments to demonstrate the feasible rates for functional compression.

\subsection{Numerical Demonstration of Rate Savings}
\label{numerical}
To illustrate the gains of the proposed encoding scheme, we numerically evaluate the example table $P_{X_1,\, X_2}$ given in Fig. \ref{fig:Functional_CI}. Using Props. \ref{Coding_GK_f(X_1_X_2)} and \ref{functional_common_information_two_sources} we compare the rates needed for zero-error encoding and decoding of $f(X_1,\, X_2)$.

\begin{figure*}[t!]
    \centering
    \includegraphics[width=0.49\textwidth]{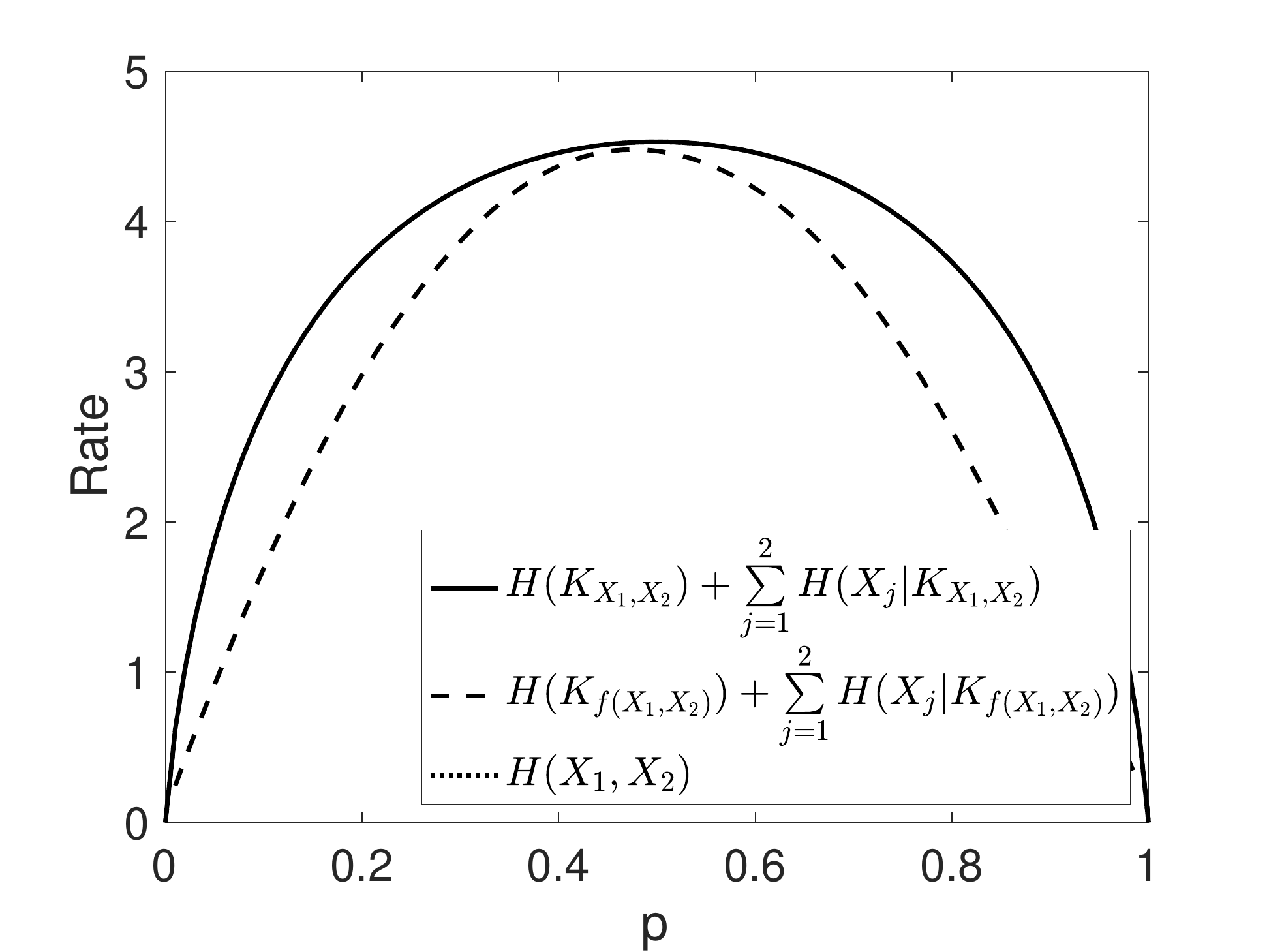}
    \includegraphics[width=0.49\textwidth]{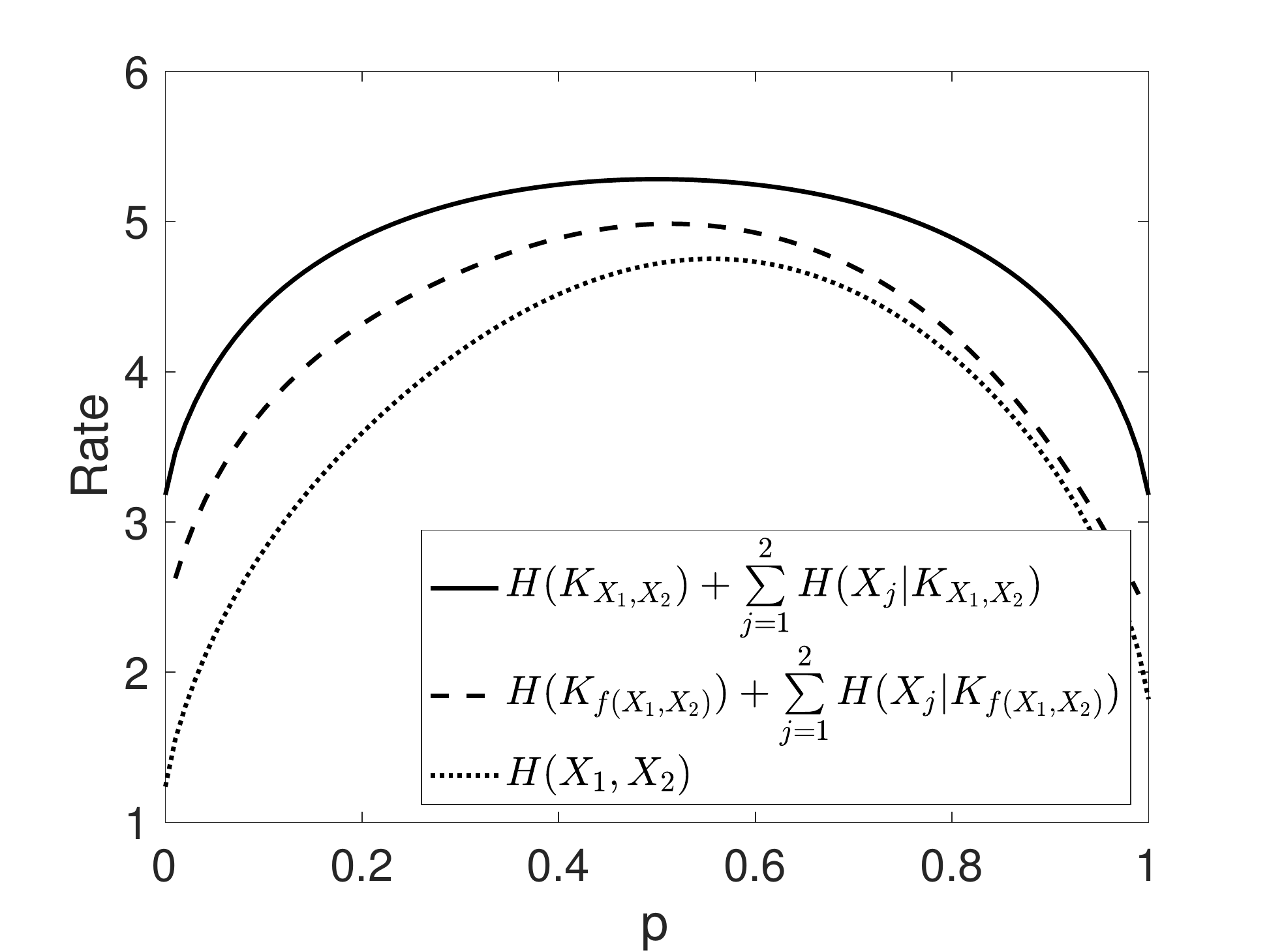}
    \includegraphics[width=0.49\textwidth]{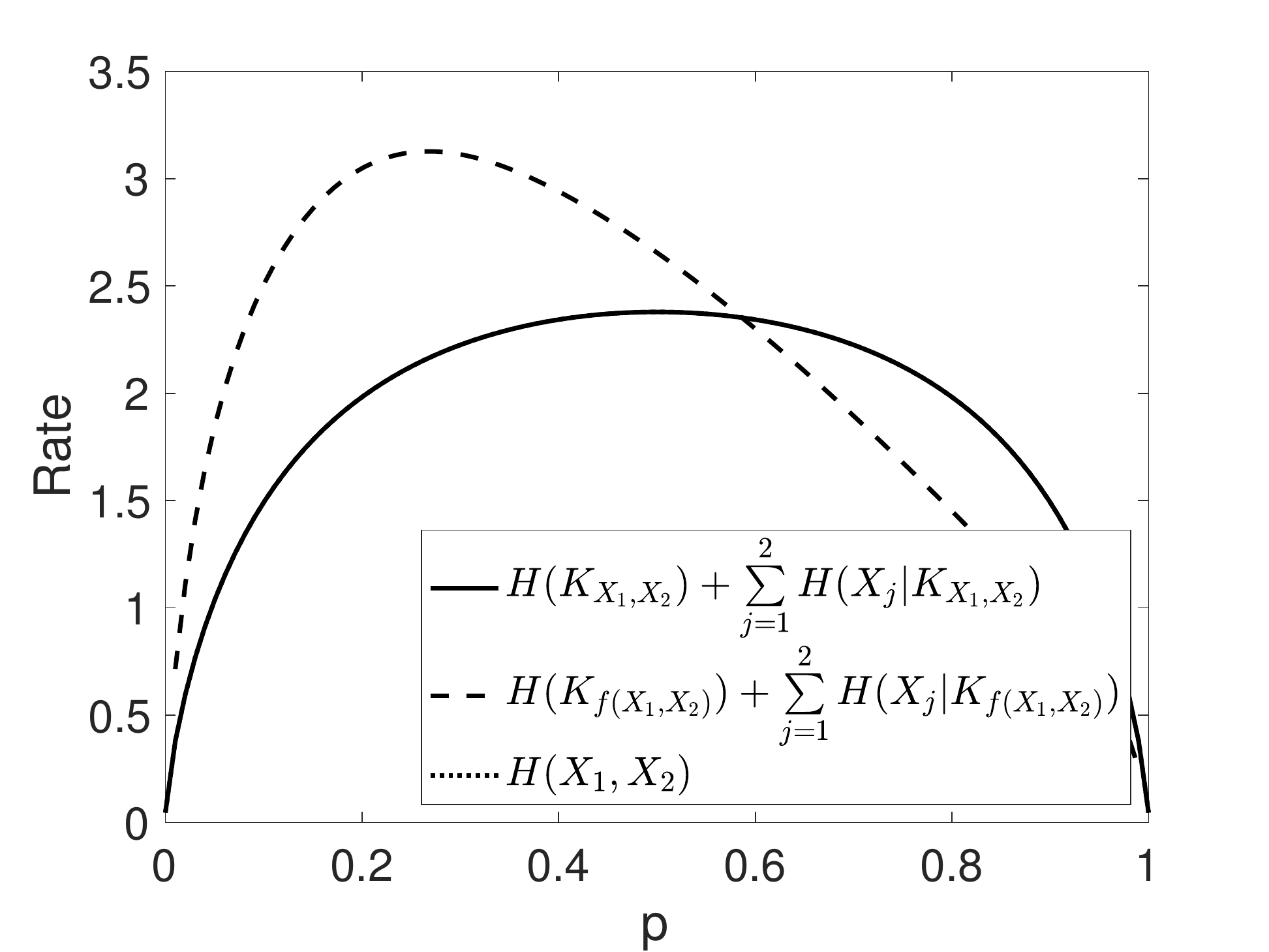}
    \caption{Rates for zero-error encoding/decoding of $X_1$ and $X_2$ in the presence of $K_{X_1,\, X_2}$ and $K_{f(X_1,\, X_2)}$ for scenarios (a) Independent sources (top left), (b) Correlated sources (top right), and (c) Independent sources with high failure rate (bottom), respectively.}
    \label{fig:Rate_X1X2}
\end{figure*}

\begin{ex}
We consider the functional common information graph as shown in Fig. \ref{fig:Functional_CI}. Given that only one source needs to transmit, if source 1 transmits, the rate required for zero-error encoding of function given $V=i$ is
\begin{align}
\label{Source_1_marginal_rates_functional_compression}
H_{G^i_{X_1}}(X_1\vert V=i)= 
\begin{cases}
\vspace{0.3cm}
h\left(\frac{\sum\limits_{n_{x_1}\in \mathcal{C}_{k1}} P_{X_1}(x_1)}{\sum\limits_{n_{x_1}\in \mathcal{C}_1} P_{X_1}(x_1)}\right),\quad\mbox{if}\quad i=1,\\
\vspace{0.3cm}
h\left(\sum\limits_{n_{x_1}\in \mathcal{C}_{k2}} P_{X_1}(x_1)\right),\quad\mbox{if}\quad i=2,\\
0,\quad\mbox{if}\quad i=3,\,4,
\end{cases}
\end{align}
where $P_{X_1}(x_1)$ denotes the probability that $X_1$ takes the value $x_1$. If on the other hand source $X_2$ transmits, then the rate required for zero-error encoding of function given $V$ is 
\begin{align}
\label{Source_2_marginal_rates_functional_compression}
H_{G^i_{X_2}}(X_2\vert V=i)= 
\begin{cases}
\vspace{0.3cm}
h\left(\sum\limits_{n_{x_2}\in \mathcal{C}_{k1}} P_{X_2}(x_2)\right),\quad \mbox{if}\quad i=1,\\
\vspace{0.3cm}
h\left(\frac{\sum\limits_{n_{x_2}\in \mathcal{C}_{k2}} P_{X_2}(x_2)}{\sum\limits_{n_{x_2}\in \mathcal{C}_{2}} P_{X_2}(x_2)}\right),\quad \mbox{if}\quad i=2,\\
0,\quad\mbox{if}\quad i=3,\,4,    
\end{cases}
\end{align}
where $P_{X_2}(x_2)$ is the probability that $X_2$ takes the value $x_2$. 

To numerically characterize the rate region for functional compression, for the example in Fig. \ref{fig:Functional_CI} we demonstrate three joint source distribution scenarios using different compression schemes, including the helper-based and the joint graph encoding schemes 
in Fig. \ref{fig:FunctionalCommonInformationRate} (as outlined in the legend shown in top-left):

{\bf (a) Independent sources.} $X_1$ and $X_2$ are independent binomially distributed variables $P_{X_1}\sim B(n_1,p)$ and $P_{X_2}\sim B(n_2,1-p)$ where $n_1=6$, $n_2=5$. In Fig. \ref{fig:FunctionalCommonInformationRate} (a) we illustrate the different rates required as a function of $p\in[0,\,1]$. Note that here $H(K_{X_1,\, X_2} )=0$, i.e., the common information graph has a single connected component, and the helper-based compression (of Prop. \ref{Coding_GK_f(X_1_X_2)}) and joint graph compression with rates $H_{G_{X_1},G_{X_2}}(X_1,\, X_2)=H_{G_{X_1}}(X_1)+H_{G_{X_2}}(X_2)$ where equality follows due to independence of $X_1$ and $X_2$ perform the same. Note also that $H(K_{f(X_1,\, X_2)} )$ is close to the fundamental lower bound, i.e., $H(V)\approx H(f(X_1,\, X_2))$, because $V$ can exploit the structure of $f(X_1,\, X_2)$, and partition $P_{X_1,\, X_2}$ into chunks via orthogonal trimming whenever possible, e.g., for extracting $\mathcal{C}_1$ and $\mathcal{C}_2$, and compresses the nontrimmable parts, e.g., $\mathcal{C}_3$ and $\mathcal{C}_4$, independently. This scheme provides significant savings in terms of functional compression over compression of data because given $V$ the marginal rate required from either source, i.e.,  (\ref{Source_1_marginal_rates_functional_compression}) or (\ref{Source_2_marginal_rates_functional_compression}) is negligible. 

{\bf (b) Correlated sources.} $(Y_1,Y_2)$ is a jointly distributed pair of Binomial random variables with distributions $P_{Y_1}\sim B(N,p)$ and $P_{Y_2}\sim B(N,1-p)$ where $N\sim \mbox{Pois}(\lambda)$ with $\lambda=5$. We determine $(X_1,\, X_2)$ by clumping together the entries of $P_{Y_1,Y_2}$ such that $P_{X_1}(6,x_2)=\sum_{y_1:y_1\geq 6} P_{Y_1,Y_2}(y_1,x_2)$ and $P_{X_2}(x_1,5)=\sum_{y_2:y_2\geq 5} P_{Y_1,Y_2}(x_1,y_2)$. In Fig. \ref{fig:FunctionalCommonInformationRate} (b) we plot the different rates required versus $p$. Similar to the scenario (a) we note that $H(K_{f(X_1,\, X_2)} )\approx H(f(X_1,\, X_2))$ as the helper can exploit the structure of $P_{X_1,\, X_2}$. 
However, $K_{X_1,\, X_2}=0$ and the scheme in Prop. \ref{Coding_GK_f(X_1_X_2)} does not provide savings over $H_{G_{X_1},G_{X_2}}(X_1,\, X_2)$ due to the correlation between $X_1$ and $X_2$. This is true especially at low $p$ such that $\mathcal{C}_2$ and $\mathcal{C}_3$ have higher probabilities and are hard to distinguish 
and $H(V)$ is high 
versus at high $p$ such that $\mathcal{C}_2$ has a higher probability and hence is easier to distinguish.  
\end{ex}

{\bf (c) Independent sources with high failure rate.} We next consider $X_1$ and $X_2$ that are independent binomially distributed random variables $P_{X_1}\sim B(n_1,p)$ and $P_{X_2}\sim B(n_2,0.001)$ where $n_1=6$, $n_2=5$. In Fig. \ref{fig:FunctionalCommonInformationRate} (c) we illustrate the rates required. Note that $K_{X_1,\, X_2}=0$ and hence Prop. \ref{Coding_GK_f(X_1_X_2)} yields the same result as $H_{G_{X_1},G_{X_2}}(X_1,\, X_2)=H_{G_{X_1}}(X_1)+H_{G_{X_2}}(X_2)$ for the independent sources. Because the failure probability for $X_2$ is very high, the first column of $X_2$ hence the graphs $\mathcal{C}_1$, $\mathcal{C}_2$, and $\mathcal{C}_4$ are more likely and $\mathcal{C}_3$ has negligible probability. 
For small $p$ our scheme requires very high rates because $X_1$ takes smaller values and $\mathcal{C}_i$'s have nonzero probabilities. As $p$ increases, $X_1$ can take larger values and the function outcome becomes deterministic which is captured by $\mathcal{C}_{22}$. 
As $p$ approaches $0.5$ the helper rate drops significantly because all outcomes of $X_1$ become equally likely and $\mathcal{C}_2$ has high probability, hence $H(V)$ gets small. 
As $p$ increases further towards $1$, $X_1$ is skewed towards high values where $\mathcal{C}_2$ has a high probability and $\mathcal{C}_1$ has very low probability.  Since the functional common information scheme needs to distinguish $\{\mathcal{C}_i\}$ given $P_{X_1,\, X_2}$, the entropy $H(V)$ may be high because it is generated based on trimming $P_{X_1,\, X_2}$ and given $P_{X_1,\, X_2}$, the index $V$ significantly depends on the marginal distribution of $X_1$ and has almost no sensitivity to $X_2$ because $P_{X_2}(X_2=0)=0.999$. Hence, it is possible that $H(V) \geq H(f(X_1,\, X_2))$. 
However, when $p$ is large enough such that $H(V)$ is small as $V=2$ with high probability, the marginal rate from either source is negligible because $P_{X_1,\, X_2}$ is skewed such that $H(\mathcal{C}\vert V)=0$, i.e., $H(K_{f(X_1,\, X_2)} )\approx H(V)$. 

We next quantify how we can exploit $K_{f(X_1,\, X_2)}$ for joint decoding of source data. We compare the rates needed for zero-error encoding and decoding of $(X_1,\, X_2)$ for scenarios (a)-(c) in Fig. \ref{fig:Rate_X1X2} using Prop. \ref{Coding_GK_f(X_1_X_2)} and Prop. \ref{zero_error_X1X2_functional_CI}. Because $X_1$ and $X_2$ are independent in (a) and (c), $K_{X_1,\, X_2}=0$ and the sum rate of the helper-based encoding is the same as the joint compression limit $H(X_1,\, X_2)$. In (b), while $K_{X_1,\, X_2}=0$, the sum rate of helper-based encoding is higher than $H(X_1,\, X_2)$ because $X_1,\,X_2$ are correlated. In (a) and (b), $K_{f(X_1,\, X_2)}$ captures the functional coupling and performs close to the fundamental bound $H(f(X_1,\, X_2))$. Hence, having the knowledge of $K_{f(X_1,\, X_2)}$ not only reduces the complexity of joint decoding of $(X_1,\, X_2)$, but also provides rate savings. 

In Sect. \ref{CompressionRates} we focus on the class of permutation invariant functions. In particular, we extract the common information for computing permutation invariant functions and efficient encoding of permutation invariance in distributed compression.

\section{Representation of Permutation Invariance} 
\label{CompressionRates}

Permutation invariant functions find applications in a broad range of domains. In supervised learning, the output label for a set is invariant or equivariant to the permutation of set elements \cite{zaheer2017deep}. Other applications include similarity search and metric learning \cite{zaheer2017deep}, personal recommendation systems \cite{zhou2007bipartite}, and GNNs \cite{zhou2018graph}, \cite{carroll2020finite}, as summarized in Sect. \ref{related}. 
To illustrate permutation invariance of the graph data $X$ for an invariant function $f(X)$, we provide an invariant graph diagram in Fig. \ref{fig:equivariant}, where rearrangement of the rows and columns of $X$ according to $g$ does not change the graph, i.e., $f(g\,\cdot X)=f(X)$ \cite{maron2018invariant}, \cite{maron2019universality}. 
The class of permutation invariant functions includes the summation, average, majority, and parity functions, among others \cite{Benso2003}. In the remainder of the paper, we restrict our attention to this class of functions. 
\begin{figure}[t!]
    \centering
    \includegraphics[width=0.45\columnwidth]{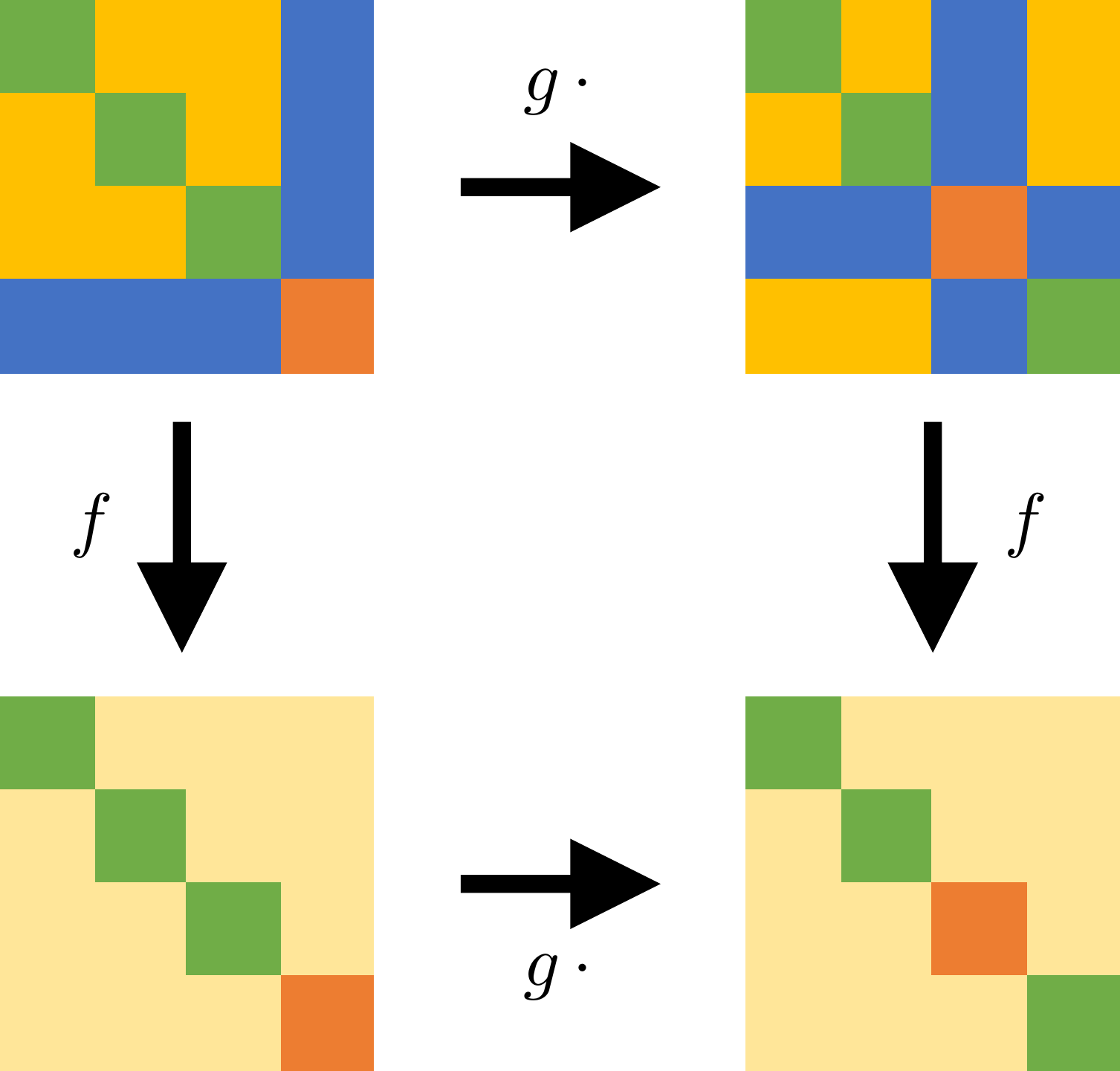}
    \caption{An equivariant commutative diagram where $g\,\cdot$ is the map that transforms graph data $X$ and returns $g\,\cdot X$.}
    \label{fig:equivariant}
\end{figure}

%

In \cite{zaheer2017deep}, authors characterized the structure of permutation invariant functions. More precisely, \cite{zaheer2017deep} gives a representation via inner and outer transformations.

\begin{theo}\label{PermInv}{\bf A representation result for permutation invariance \cite{zaheer2017deep}.} A function $f(S)$ operating on a set $S= \{x_1, \dots , x_M\}$ having elements from a countable universe, is a valid set function, i.e., invariant to the permutation of elements in $X$, if and only if it decomposes as
\begin{align}
\Phi\Big(\sum\limits_{x\in S} \psi(x)\Big)    
\end{align}
for suitable transformations $\psi: \mathbb{R}\to \mathbb{R}^{m}$ and $\Phi:\mathbb{R}^{m}\to \mathbb{R}$ where $m\leq 2M+1$, where $\psi(x)$ can be a vector-valued function. Hence, the input of the outer function $\Phi$ can be a vector.
\end{theo}

In this paper, we consider two input variables $x_1$ and $x_2$ that are sampled from the same discrete alphabet. Hence $S=\{x_1,x_2\}$. We let $X_1=\psi(x_1)$ and $X_2=\psi(x_2)$ be suitable inner transformations implemented at each source independently. We call the transformed variables $X_1$ and $X_2$ the source variables. The generalization to multiple sources is standard.

From Theorem \ref{PermInv}, by letting $X_1=\psi(x_1)$ and $X_2=\psi(x_2)$ a permutation invariant function of two variables can be represented as 
\begin{align}
\label{SymmetricFunction}
f(X_1,\, X_2)=\Phi(X_1+X_2),
\end{align}
where the range $\mathcal{X}$ of $\psi$ is a discrete space, e.g., in classification. A continuous space can be quantized as well, e.g., in case of regression, for a generalized rate-distortion scenario.


From data processing, the rate of computing (\ref{SymmetricFunction}) is upper bounded by the sum rate of $X_1+X_2$. In other words, if $X_1+X_2$ can be computed, then any $f(X_1,\, X_2)$ that satisfies (\ref{SymmetricFunction}) can be computed. Therefore, the rate region to compute the sum is encompassed by the rate region for computing $\Phi(X_1+X_2)$.

To motivate our approach we start with a canonical example.  
Let each source be uniform over $\mathcal{X}=\{0,\,1,\,\dots,\,n-1\}$. Exploiting the notion of graph entropy to compute $X_1+X_2$, since each source outcome needs to be distinguished, the sum rate needed to compute $X_1+X_2$ is given as 
\begin{align}
\label{R_S}
R_S=2\log(n),
\end{align}
which is an upper bound to the total rate needed to compute (\ref{SymmetricFunction}) because (\ref{R_S}) models the rate without exploiting the common information between the source characteristic graphs. In Sect. \ref{Bipartitions_with_helpers} we generalize this setting to nonuniform source distributions.

Permutation invariant function computing can be posed as a bipartite graph compression problem. To that end, we next construct a bipartite graph $G_f=(U,V,E)$ from the characteristic graphs $G_{X_1}$ and $G_{X_2}$ the individual sources build independently to compute the symmetric function $f$, whose partition has the parts $U$ and $V$ which correspond to the collection of equivalence classes in $G_{X_1}$ and $G_{X_2}$, respectively, and $E$ denotes the edges of the graph which capture the correlation between these 2 characteristic graphs. More specifically, the bipartite graph $G_f$ has the following properties:
\begin{enumerate}[i.]
    \item The set of nodes $U$ and $V$ that partition $G_f$ are disjoint.
    \item The set of nodes $U$ and $V$ correspond to the set of equivalence classes in $G_{X_1}$ and $G_{X_2}$, respectively. The equivalence class of $G_{X_1}$ corresponding to a function outcome $a\in  \operatorname {Range} (f)$, where $\operatorname {Range} (f)$ denotes the range\footnote{The range of a function is the set of the images of all elements in the domain \cite{trench2013introduction}.} of $f$. 
    is $u=\{x_1\in \mathcal{X}\vert f(x_1,\cdot)=a\}\in U$. Similarly, the equivalence class of $G_{X_2}$ given a function outcome $a\in  \operatorname {Range} (f)$ is determined as $v=\{x_2\in \mathcal{X}\vert f(\cdot,\, x_2)=a\}\in V$. 
    \item $G_f$ is a balanced bipartite graph with $|U|=|V|$, i.e., the two subsets of nodes have equal cardinality because $f$ is permutation invariant.
    \item There is an edge between node $u\in U$ and node $v\in V$, i.e., $(u,\,v)\in E$, if and only if $P_{X_1,\, X_2}(X_1\in u ,\, X_2\in v) >0$. Note however that $(u,\,v)\in E$ does not imply $(v,\, u)\in E$ for $v\in U$ and $u\in V$.
    \item If $u \in U$ and $v\in V$ are connected, i.e., $(u,\, v)\in E$, and similarly  $(v,\, u)\in E$, then both edges yield the same function outcome. However, this is not necessary as the source alphabets and distributions of the characteristic graphs $G_{X_i}$ might differ, and the edges do not imply each other.
    \item If the maximal independent set of $G_{X_1}$ and $G_{X_2}$ is a singleton, then $u$ and $v$ representing these sets correspond to $n_{x_1}$ and $n_{x_2}$ as defined in Sect. \ref{GraphCompression}, respectively.
\end{enumerate}

In distributed compression, source $X_1$ only has the knowledge of $U$, and similarly source $X_2$ only has the knowledge of $V$. However, neither has the knowledge of the set of edges $E$ in $G_f$. Note that $E$ is determined by the function $f(X_1,\, X_2)$ as well as $P_{X_1,\, X_2}$. 

If $G_f$ is complete, it has $|U|^2$ edges and up to $n+\frac{n(n-1)}{2}$ distinct function outcomes to be distinguished. On the other hand  if $G_f$ is not connected, it may have more than one bipartition (connected component) \cite{chartrand2019chromatic}. In that case, encoding the function is easy once the right bipartition is extracted. 
To that end, we employ a helper that exploits the structure of $E$ to decompose $G_f$ into its bipartitions. Encoding the bipartitions can provide rate savings in compression. 
The helper notion is linked to the G{\'a}cs-K{\"o}rner common information $K_{X_1,\, X_2}$ between $X_1$ and $X_2$ that captures the combinatorial structure of $P_{X_1,\, X_2}$, 
and $K_{X_1,\, X_2}$ only needs to be transmitted once because it can be separately extracted from either source \cite{gacs1973common}. 

We next quantify how much common randomness the helper can extract for a permutation invariant function, and demonstrate that it is not less than $K_{X_1,\, X_2}$. In $G_f$, there is an edge between $u$ and $v$ provided that condition (iv) is satisfied. 
Furthermore, some edges in $G_f$ might have identical function outcomes because permutation invariance induces a symmetry and/or the function might be non-surjective. 
Helper can capture this to refine the combinatorial structure in $K_{X_1,\, X_2}$, and provide more effective encoding and transmission of data.


\section{Efficient Partitioning of Bipartite Graph $G_f$ for Effective Distributed Compression of Permutation Invariant Functions}
\label{Bipartitions_with_helpers}

In this section, to understand the fundamental limits, we investigate the achievable rates of distributed computing of permutation invariant functions by incorporating a helper or a common information variable. In particular, we consider three different approaches: (i) exploiting the bipartitions via variable $K_B$, (ii) inferring the low probability edges via variable $K_{\delta}$, and (iii) exploiting the structure brought by the function via variable $K_S$. We next detail each approach.

\subsection{Exploiting Bipartitions in $G_f$}
\label{sect:helper_bipartitions}

To motivate our approach we start with a simple example for joint compression of $G_f$. In our compression scheme each source variable is uniformly distributed over $\mathcal{X}=\{0,\,1,\,\dots,\,n-1\}$. Hence, $G_f$ has $2n$ nodes in total, i.e., $n=|U|=|V|$. We denote the set of bipartitions by $\mathcal{K}$, the cardinality (count) of the bipartitions by $|\mathcal{K}|$, and bipartition $k\in \mathcal{K}$ by $\mathcal{G}_k=(U_k,\,V_k,\,E_k)$. Let $p_k$ be the probability that $\mathcal{G}_k$ in $G_f$ is selected, and $n_k$ be the size of nodes in $U_k$, i.e., $\sum_k n_k=n$. The rate required to specify the index of the bipartition, i.e., $K_B$, satisfies
\begin{align}
\label{common}
H(K_B)=-\sum\limits_{k\in\mathcal{K}} p_k\log p_k,
\end{align}
which is the G{\'a}cs-K{\"o}rner common information between $G_{X_1}$ and $G_{X_2}$ which embed the source equivalence classes. Note that $K_B$ is not necessarily identical to $K_{X_1,\, X_2}$ unless each maximal independent set in any of $G_{X_1}$ and $G_{X_2}$ is a singleton. Because of aggregation, it is also true that $H(K_B)\leq H(K_{X_1,\, X_2})$.

We consider different scenarios to demonstrate the gains of joint compression where we exploit the common randomness between the partitions of $G_f$ through a helper as a proxy for establishing bipartitions to provide rate savings. 
The helper encodes the index of the bipartition using $H(K_B)$ bits. Once the index is known, sources need $\log(n_k)$ bits each to encode $U_k$ and $V_k$ assuming that the outcome is uniform over bipartition $\mathcal{G}_k$. Hence, the total rate required to compute $X_1+X_2$ via exploiting the common randomness is given as
\begin{align}
\label{R_J}
R_J = 2\sum\limits_{k\in\mathcal{K}} p_k \log n_k+H(K_B).   
\end{align}

\begin{prop}
For permutation invariant functions with uniform source alphabet, the rate saving $R_{\Delta}=R_S-R_J$ of the common information model satisfies  
\begin{align}
\label{rate_comparison}
R_{\Delta}=R_S-R_J \leq  H(K_B).   
\end{align}

\end{prop}

\begin{proof}
The log sum inequality states that $$\sum\limits_j a_j\log \frac{a_j}{b_j} \geq a \log \frac{a}{b},$$ where $a_j\geq 0$ and $b_j\geq 0$ such that $\sum_j a_j=a$ and $\sum_j b_j=b$ \cite{cover2012elements}. The equality is if and only if $\frac  {a_{j}}{b_{j}}$ are equal for all $j$, i.e., $a_{j}=cb_{j}$ for all $j$.
Exploiting this inequality we infer that $$\log n \leq \sum\limits_{k\in\mathcal{K}} p_k \log(n_k)+H(K_B).$$ Multiplying each side with $2$, rearranging the terms, and using (\ref{rate_comparison}) we get the result.
\end{proof}

From (\ref{rate_comparison}) the savings can be more eminent when $H(K_B)$ is large. 
If the bipartition size is uniform, i.e., $n_k=\frac{n}{|\mathcal{K}|}$, and probability $p_k\sim \textrm{Zipf}(\gamma)$ where $\gamma$ denotes its skew, $H(K_B)\leq \log(|\mathcal{K}|)$ where equality is if and only if $\gamma=0$. From (\ref{R_J}), $R_J\leq 2\log(n)-\log(|\mathcal{K}|)$. Hence, the gain $R_{\Delta}\geq \log(|\mathcal{K}|)$ improves with $|\mathcal{K}|$ and skewness. If the bipartition size $n_k$ is arbitrary, 
\begin{align}
R_J&\overset{(a)}{\leq} 2\sum\limits_{k\in\mathcal{K}} p_k \log(n_k)+\log(|\mathcal{K}|)\nonumber\\
&\overset{(b)}{\leq} 2\log\Big(\sum\limits_{k\in\mathcal{K}} p_k n_k\Big)+\log(|\mathcal{K}|), \nonumber
\end{align}
where $(a)$ is equality if and only if $\gamma=0$ and $(b)$ follows from concavity of log. 
It is possible that $R_{\Delta}\geq\log(|\mathcal{K}|)$ when partitions with small $n_k$ have large $p_k$. Hence, having different size partitions $n_k$, larger $|\mathcal{K}|$ and $\gamma$ improves the gain $R_{\Delta}$.

In Fig. \ref{fig:gain_permutation_invariant_CI} we illustrate the gain of joint compression of $f(X_1,X_2)=X_1+X_2$  for several scenarios, where the partition size could be uniform, i.e., $n_k=\frac{n}{|\mathcal{K}|}$, or arbitrary such that $\sum\limits_{k\in\mathcal{K}} n_k=n$. We note that $\gamma$ captures the skewness of the distribution on the partition indices, where partition $k\in\mathcal{K}$ has probability $p_k\sim \textrm{Zipf}(\gamma)$. 
\begin{figure}[t!]
    \centering
    \includegraphics[width=0.5\columnwidth]{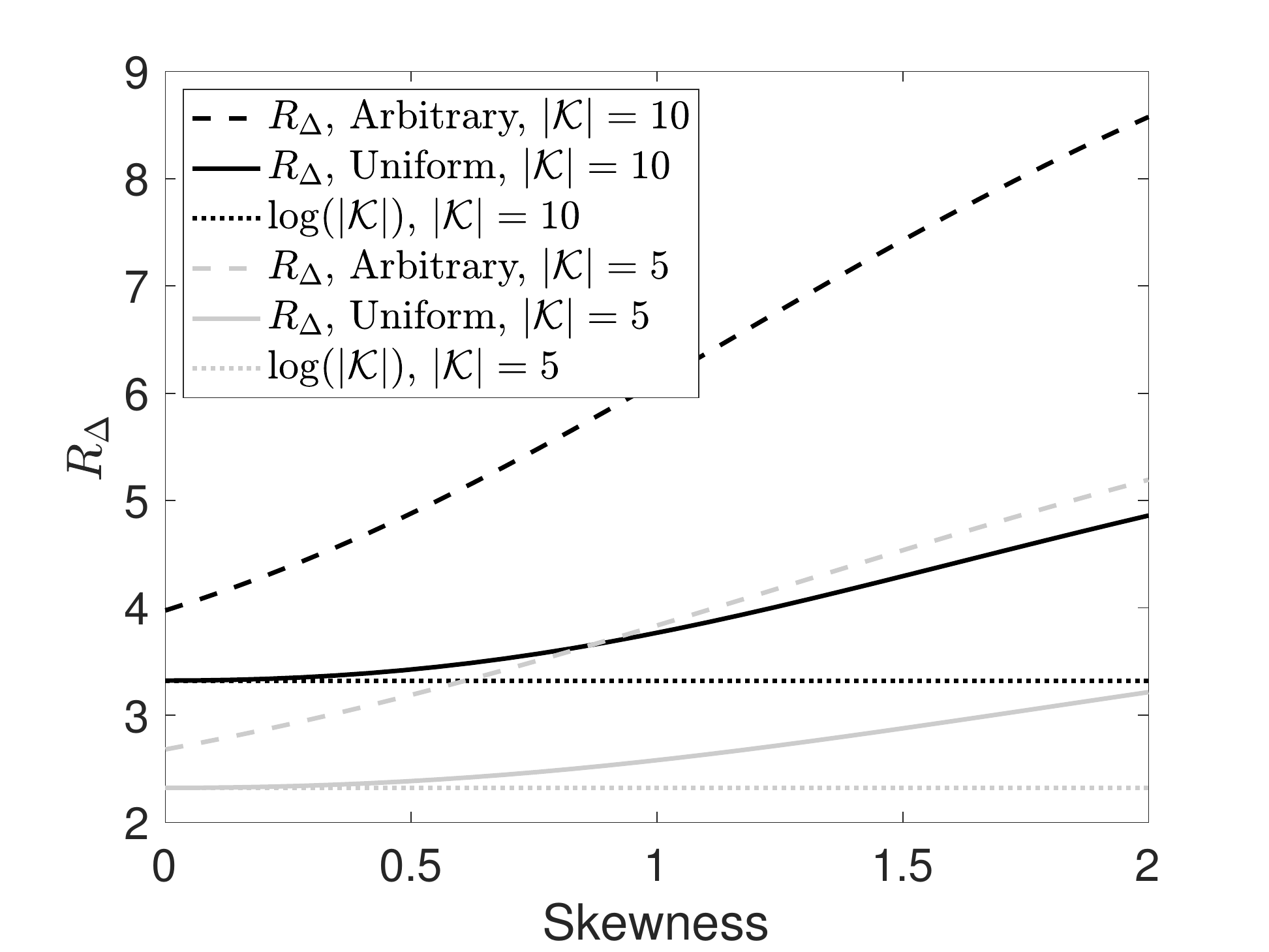}
    \caption{The rate savings in (\ref{rate_comparison}) of joint compression versus skewness $\gamma$ for different scenarios.}
    \label{fig:gain_permutation_invariant_CI}
\end{figure}

\begin{prop}\label{RateRegionPermutationInvariant}
The rate region for effective distributed compression of permutation invariant functions $f(X_1,\, X_2)$ with a helper is characterized by
\begin{align}
R_J\geq H(K_B)+H_{G_{X_1}}(X_1\vert K_B)+H_{G_{X_2}}(X_2\vert X_1),\nonumber    
\end{align}
where $H(K_B)$ is the asymptotic rate of the helper, $H_{G_{X_1}}(X_1\vert K_B)$ is the rate at which source $X_1$ sends to identify its equivalence class. The term $H_{G_{X_2}}(X_2\vert X_1)$ denotes the rate of source $X_2$ to identify the function outcome given the equivalence class of $X_1$. Note that $H_{G_{X_2}}(X_2\vert X_1)$ might be zero given $K_B$ if the bipartitions of $G_f$ are sparsely connected within themselves. We can swap the roles of $X_1$ and $X_2$.
\end{prop}

We next consider a bipartite graph construction for computing a permutation invariant function of correlated variables, where source alphabets are identical and denoted by $\mathcal{X}$.
\begin{ex}  \label{SymmetricSum}
Sum function is permutation invariant. Let $f(X_1,\, X_2)=X_1+X_2$ where both sources have alphabets $\mathcal{X}=\{0,\,1,\,2\}$. The joint probability matrix $P_{X_1,\, X_2}$ with entries ordered in an increasing fashion and the table of function outcomes $F$ are given as
\begin{align}
P_{X_1,\, X_2} = \frac{1}{24}\begin{bmatrix} 12 & 0 & 0 \\ 0 & 2 & 3 \\ 0 & 3 & 4\end{bmatrix},\quad F= \begin{bmatrix} 0 & X & X \\ X & 2 & 3 \\ X & 3 & 4 \end{bmatrix}.\nonumber
\end{align}
Note that for given $P_{X_1,\, X_2}$ which is symmetric, we have $H(X_1)=H(X_2)=h(\frac{1}{2},\,\frac{5}{24},\,\frac{7}{24})=1+\frac{1}{2}\cdot h\Big(\frac{5}{12}\Big)=1.4899$ and 
$H(X_2\vert X_1)=\frac{1}{2}\cdot 0+\frac{5}{24}\cdot h(\frac{2}{5})+\frac{7}{24}\cdot h(\frac{3}{7})$, and $$\mathcal{R}_{(X_1,\, X_2)}=H(X_1,\, X_2)=1.9796.$$ 
In this example, $f(X_1,\, X_2)\in\{0,\,2,\,3,\,4\}$ where the probabilities at the respective coordinates are $P_{f(X_1,\, X_2)}=\Big(\frac{1}{2},\,\frac{1}{12},\,\frac{1}{4},\,\frac{1}{6}\Big)$. 
Hence, the entropy of the function from the given joint distribution $P_{X_1,\, X_2}$ and the outcomes table $F$ is $$\mathcal{R}_{f(X_1,\, X_2)}=H(f(X_1,\, X_2))=1.7296,$$ which is less than the sum rate $\mathcal{R}_{(X_1,\, X_2)}$ needed to recover $(X_1,\, X_2)$. 
The bipartite graph $G_f$ for this example has 2 bipartitions and is shown in Fig. \ref{fig:Sum_CI_toy2}. 

The helper-based scheme requires $H(K_B)=H(K_{X_1,\, X_2})=1$ since both bipartitions have equal probabilities. With entropy coding, the total rate for recovering $(X_1,\, X_2)$ is 
\begin{align}
\mathcal{R}_{(X_1,\, X_2)}^H&=H(K_B)+H(X_1\vert K_B)+H(X_2\vert X_1)\nonumber\\
&=1+\Big(\frac{1}{2}\cdot 0+\frac{1}{2}\cdot h\left(\frac{5}{12}\right)\Big)+\Big(\frac{1}{2}\cdot 0+\frac{5}{24}\cdot h\left(\frac{2}{5}\right)+\frac{7}{24}\cdot h\left(\frac{3}{7}\right)\Big)
=1.9796.\nonumber    
\end{align}

\begin{figure}[t!]
    \centering
    \includegraphics[width=0.7\columnwidth]{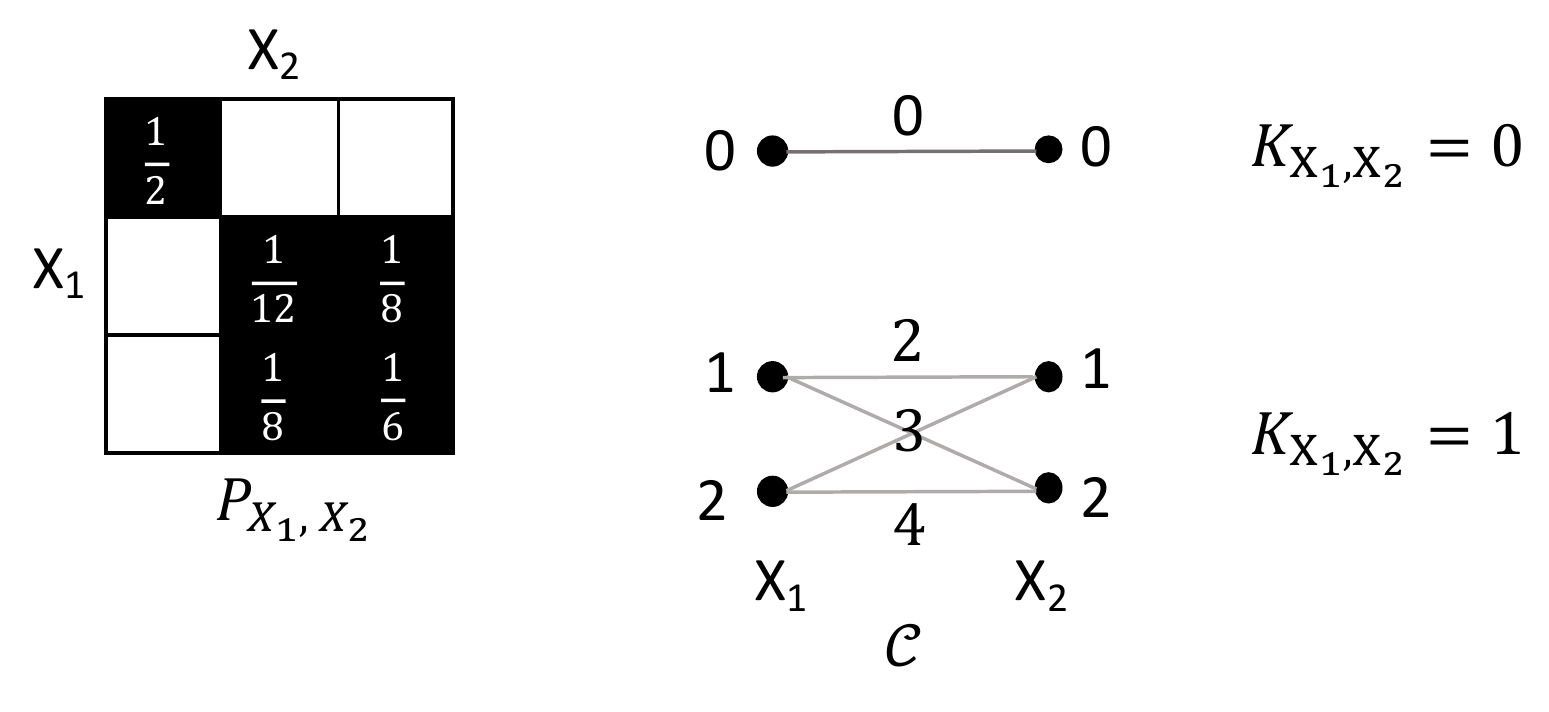}
    \caption{Computing sum function with common information. (Left) The joint distribution table. (Right) Bipartite graph $G_f$ for Example \ref{SymmetricSum}). Note that the function outcomes are indicated on the edges.}
    \label{fig:Sum_CI_toy2}
\end{figure}

The graph coloring approach requires an average number of colors to represent the sum function from each source $l=1,\,2$ which is 
$H_{G_{X_l}}(X_l)=h\Big(\frac{1}{2},\,\frac{5}{24},\,\frac{7}{24}\Big)=1.4899$. Similarly, due to the symmetry $H_{G_{X_l}}(X_l\vert K_B)=0.4899$. The total rate of the graph coloring approach is $$\mathcal{R}_{f(X_1,\, X_2)}^{\chi}=H_{G_{X_1}}(X_1)+H_{G_{X_2}}(X_2\vert X_1)=h\Big(\frac{1}{2},\,\frac{5}{24},\,\frac{7}{24}\Big)+\Big(\frac{1}{2}\cdot 0 + \frac{1}{2}\cdot h\Big(\frac{5}{12}\Big)\Big)=1.9799.$$

Using the graph coloring approach, if the source characteristic graphs can be jointly compressed, the sum rate needed for computing $f(X_1,\, X_2)$ satisfies
\begin{align}
H(K_B)+H_{G_{X_1,\, X_2}}(X_1,\, X_2\vert K_B)&=1+ \Big(\frac{1}{2}\cdot 0+\frac{1}{2}\cdot h\left(\frac{1}{6},\,\frac{1}{2},\,\frac{1}{3}\right)\Big)\nonumber\\
&=1.7296=H(f(X_1,\, X_2)).\nonumber    
\end{align}
However, the above rate is not possible to attain even with the helper because given $K_B=1$ the graph, which corresponds to the second bipartition in Fig. \ref{fig:Sum_CI_toy2}, is not block-diagonal.

In the distributed implementation, the sum rate of the helper-based model for recovering $f(X_1,\, X_2)$ is given by the following expression: 
\begin{align}
\mathcal{R}_{f(X_1,\, X_2)}^H&=H(K_B)+H_{G_{X_1}}(X_1\vert K_B)+H_{G_{X_2}}(X_2\vert X_1)\nonumber\\
&=1+\Big(\frac{1}{2}\cdot 0+\frac{1}{2}\cdot h\Big(\frac{5}{12}\Big)\Big)+\Big(\frac{1}{2}\cdot 0+\frac{5}{24}\cdot h\Big(\frac{2}{5}\Big)+\frac{7}{24}\cdot h\Big(\frac{3}{7}\Big)\Big)=1.9796.\nonumber  
\end{align}

In this example $\mathcal{R}_{f(X_1,\, X_2)}^{\chi}>\mathcal{R}_{f(X_1,\, X_2)}^H=\mathcal{R}_{(X_1,\, X_2)}^H=\mathcal{R}_{(X_1,\, X_2)}>\mathcal{R}_{f(X_1,\, X_2)}$, i.e., neither the helper nor graph coloring-based compression has an advantage over joint compression of sources since each distinct pair $(X_1,\, X_2)$ needs to be distinguished.
\end{ex}

We next consider an example bipartite graph construction for computing sum function of two variables where source alphabets are different but not disjoint. 
\begin{ex}\label{AsymmetricSum} Consider $f(X_1,\, X_2)=X_1+X_2$ where the source alphabets are given as $\mathcal{X}_1=\{1,\,2,\,3\}$ and $\mathcal{X}_2=\{0,\,1,\,3\}$. The joint probability matrix $P_{X_1,\, X_2}$ with ordered entries and the table of function outcomes $F$ are given as
\begin{align}
P_{X_1,\, X_2} = \frac{1}{4}\begin{bmatrix} 2 & 0 & 0 \\ 0 & 0 & 1 \\ 0 & 1 & 0\end{bmatrix},\quad F= \begin{bmatrix} 1 & X & X \\ X & X & 5 \\ X & 4 & X \end{bmatrix},\nonumber
\end{align}
where the function outcome is denoted by $\times$ if there is no edge between the nodes.

\begin{figure}[t!]
    \centering
    \includegraphics[width=0.7\columnwidth]{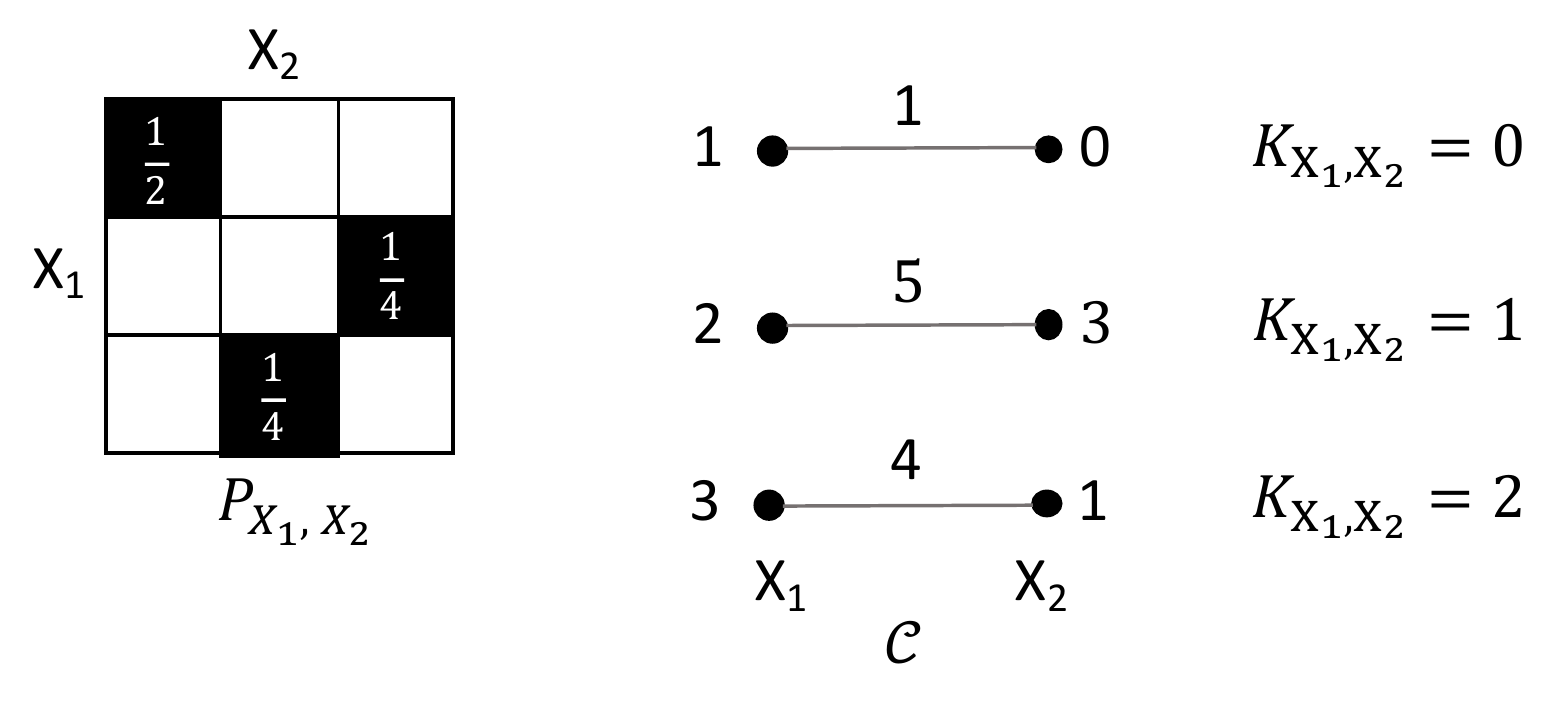}
    \caption{Computing sum function with common information. (Left) The joint distribution table. (Right) Bipartite graph $G_f$ for Example \ref{AsymmetricSum}). The function outcomes are indicated on the edges.}
    \label{fig:Sum_CI_toy1}
\end{figure}

Note that $H(X_1)=H(X_2)=h\Big(\frac{1}{2},\,\frac{1}{4},\,\frac{1}{4}\Big)=\frac{3}{2}$ and $H(X_1\vert X_2)=0$. Hence, with entropy coding, 
the total rate of the helper-based scheme for recovering $(X_1,\, X_2)$ is 
$$\mathcal{R}_{(X_1,\, X_2)}=H(X_1,\, X_2)=\frac{3}{2}.$$

In this example, $f(X_1,\, X_2)\in\{1,\,4,\,5\}$ where the probabilities at the respective coordinates are $P_{f(X_1,\, X_2)}=\Big(\frac{1}{2},\,\frac{1}{4},\,\frac{1}{4}\Big)$. 
Hence, the entropy of the function from the given joint distribution $P_{X_1,\, X_2}$ and the outcomes table $F$ is $$\mathcal{R}_{f(X_1,\, X_2)}=H(f(X_1,\, X_2))=\frac{3}{2},$$ which is the same as the sum rate $\mathcal{R}_{(X_1,\, X_2)}$ needed to recover $(X_1,\, X_2)$. 
The bipartite graph $G_f$ for this example has 3 bipartitions and is shown in Fig. \ref{fig:Sum_CI_toy1}. 

The helper-based scheme requires $H(K_B)=H(K_{X_1,\, X_2})=\frac{3}{2}$ since the bipartitions have a distribution $\Big(\frac{1}{2},\,\frac{1}{4},\,\frac{1}{4}\Big)$. With entropy coding, the total rate for recovering $(X_1,\, X_2)$ is 
\begin{align}
\mathcal{R}_{(X_1,\, X_2)}^H=H(K_B)+H(X_1\vert K_B)=\frac{3}{2}+0=\frac{3}{2}.\nonumber
\end{align}
Hence, $\mathcal{R}_{(X_1,\, X_2)}^H$ equals the sum rate needed to compute $f(X_1,\, X_2)=X_1+X_2$. Therefore, the fundamental lower bound is achieved.

The average number of colors needed to represent the sum function from each source $l=1,\,2$ is $H_{G_{X_l}}(X_l)=\frac{3}{2}$. Due to symmetry,  $H_{G_{X_l}}(X_l\vert K_B)=\frac{1}{2}\cdot 0+\frac{1}{2}\cdot 1=\frac{1}{2}$. Hence, the total rate of the graph coloring approach is
\begin{align}
\mathcal{R}_{f(X_1,\, X_2)}^{\chi}=H_{G_{X_1}}(X_1)+H_{G_{X_2}}(X_2\vert X_1)=\frac{3}{2}+0=\frac{3}{2}. \nonumber   
\end{align}

The sum rate of the helper-based model for recovering $f(X_1,\,X_2)$ is
\begin{align}
\mathcal{R}_{f(X_1,\, X_2)}^H=H(K_B)+H_{G_{X_1}}(X_1\vert K_B)+H_{G_{X_2}}(X_2\vert X_1)=\frac{3}{2}+0+0=\frac{3}{2}.\nonumber
\end{align}

In this example $\mathcal{R}_{f(X_1,\, X_2)}^H=\mathcal{R}_{f(X_1,\, X_2)}^{\chi}=\mathcal{R}_{(X_1,\, X_2)}^H=\mathcal{R}_{(X_1,\, X_2)}=\mathcal{R}_{f(X_1,\, X_2)}$. Note that while sum is a permutation invariant function, the source alphabets are not identical and each possible source pair generates a distinct outcome. The helper cannot provide any savings because $H(X_1\vert X_2)=0$ and that from the definition of GKW it holds $H(K_B\vert X_1)=H(K_B\vert X_2)=0$. As a result, there is a bijection from 
$K_B$ to 
$X_1$ given $X_2$.
\end{ex}  

In the helper-based scheme, provided that the sources are coupled, a bipartition of $G_f$ either has a unique $\Phi(X_1+X_2)$ value, where transmission from only one source suffices to compute $\Phi(X_1+X_2)$, or multiple outcomes, where the helper can still exploit the coupling between the sources and provide rate savings.

Examples \ref{SymmetricSum} and \ref{AsymmetricSum} demonstrate the power of a helper that exploits the bipartitions in $G_f$ to compute a symmetric function $f(X_1,\,X_2)$.

\subsection{Exploiting Low Probability Edges in $G_f$}
\label{sect:helper_low_prob}

A helper-based scheme could still be practical even when $G_f$ is connected, i.e., there is a single bipartition, unlike the approach in Sect. \ref{sect:helper_bipartitions} that relies on bipartitioning. If the helper can leverage the low probability edges in $G_f$, and provide the necessary rate to distinguish those edges, then the task of each source is to send a refinement to identify the outcome of the function. We next consider an example where $G_f$ is with one bipartition and a helper can leverage the low probability edges each with probability $\delta$. To capture the helper or the common information variable that can distinguish the low probability edges, we use the notation $K_{\delta}$. Note that the variable $K_{\delta}$ is different from $K_B$ and $K_{X_1,\, X_2}$, which denote the indices of the bipartite graphs $G_f$  corresponding to the functions $f(X_1,\, X_2)$ and $(X_1,\, X_2)$, respectively.

\begin{ex} \label{SymmetricLowError}
Consider the joint probability matrix $P_{X_1,\, X_2}$ and the table of permutation invariant function outcomes $F$ given below:
\begin{align}
P = \frac{1}{3} \begin{bmatrix} 0 & 1 & 0 \\ 1-\delta & 0 & \delta\\ 0 & \delta & 1-\delta\end{bmatrix},\quad F= \begin{bmatrix} X & 1 & X \\ 1 & X & 3\\ X & 3 & 2\end{bmatrix},\nonumber
\end{align}
where $\delta\in (0,\, 1)$ denotes the cross-over probability between the two bipartitions. The bipartite graph $G_f$ is shown in Fig. \ref{BipartiteGraphs_LowDelta} where the vertices are listed in the presented order. 

\begin{figure}[t!]
    \centering
    \includegraphics[width=0.25\columnwidth]{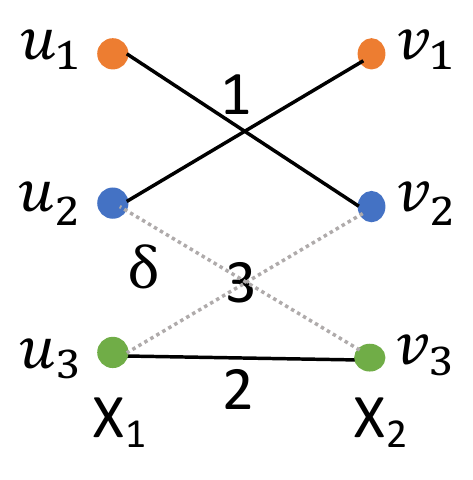}
    \caption{Bipartite graph $G_f$ for Example \ref{SymmetricLowError}, where each equivalence class (node) in either characteristic graph is denoted by a different color, each dotted or dashed edge has a probability $\delta$ and the function outcomes are indicated on the edges.}\label{BipartiteGraphs_LowDelta}
\end{figure}

The entropy of the function from the given joint distribution $P_{X_1,\,X_2}$ and the outcomes table $F$ is $H(f(X_1,\,X_2)) = h\left(\frac{2-\delta}{3},\, \frac{1-\delta}{3},\, \frac{2\delta}{3}\right)$. Note that when $\delta=0$, the source distributions are uniform. Hence, a rate upper bound not exploiting the low probability edges in $G_f$ is $$H_{G_{X_1}}(X_1)+H_{G_{X_2}}(X_2) = 2\log(3).$$ The sum rate required from the graph coloring approach to compute the function in the case of no helper is
\begin{align}
\mathcal{R}_{f(X_1,\, X_2)}^{\chi}=H_{G_{X_1}}(X_1)+H_{G_{X_2}}(X_2\vert X_1) = \log(3)+\Big(\frac{1}{3}\cdot 0+\frac{1}{3}\cdot h(\delta)+\frac{1}{3}\cdot h(\delta)\Big).\nonumber
\end{align}
We next consider a helper-based scheme. Assuming that these edges between the bipartitions are not present, $G_f$ has 2 bipartitions. Hence, $\delta$ models the cross-over probability between the two bipartitions. The helper's role is to identify whether or not the edges  between the bipartitions of $G_f$ with probability $\delta$ exist. This requires the helper to transmit $H(K_{\delta})=2h\left(\frac{\delta}{3}\right)$ bits per use asymptotically where the scaling $2$ is due to symmetry of the edges. 
Hence, the sum rate required to compute the function with a helper that exploits the common information between the bipartitions in $G_f$ is given as
\begin{align}
\mathcal{R}_{f(X_1,\, X_2)}^H=H(K_{\delta})+H_{G_{X_1}}(X_1\vert K_{\delta}) = 2h\left(\frac{\delta}{3}\right)+h\left(\frac{2}{3}\right).\nonumber
\end{align}
\end{ex}
We compare these rates of entropy coding and helper-based schemes in Fig. \ref{fig:BipartiteGraph_LowDelta}. 
The benefit of helper in functional compression is eminent when $\delta\in (0,\, 0.16)$, i.e., $\delta$ is low. 

\begin{figure*}[t!]
    \centering
    \includegraphics[width=0.49\textwidth]{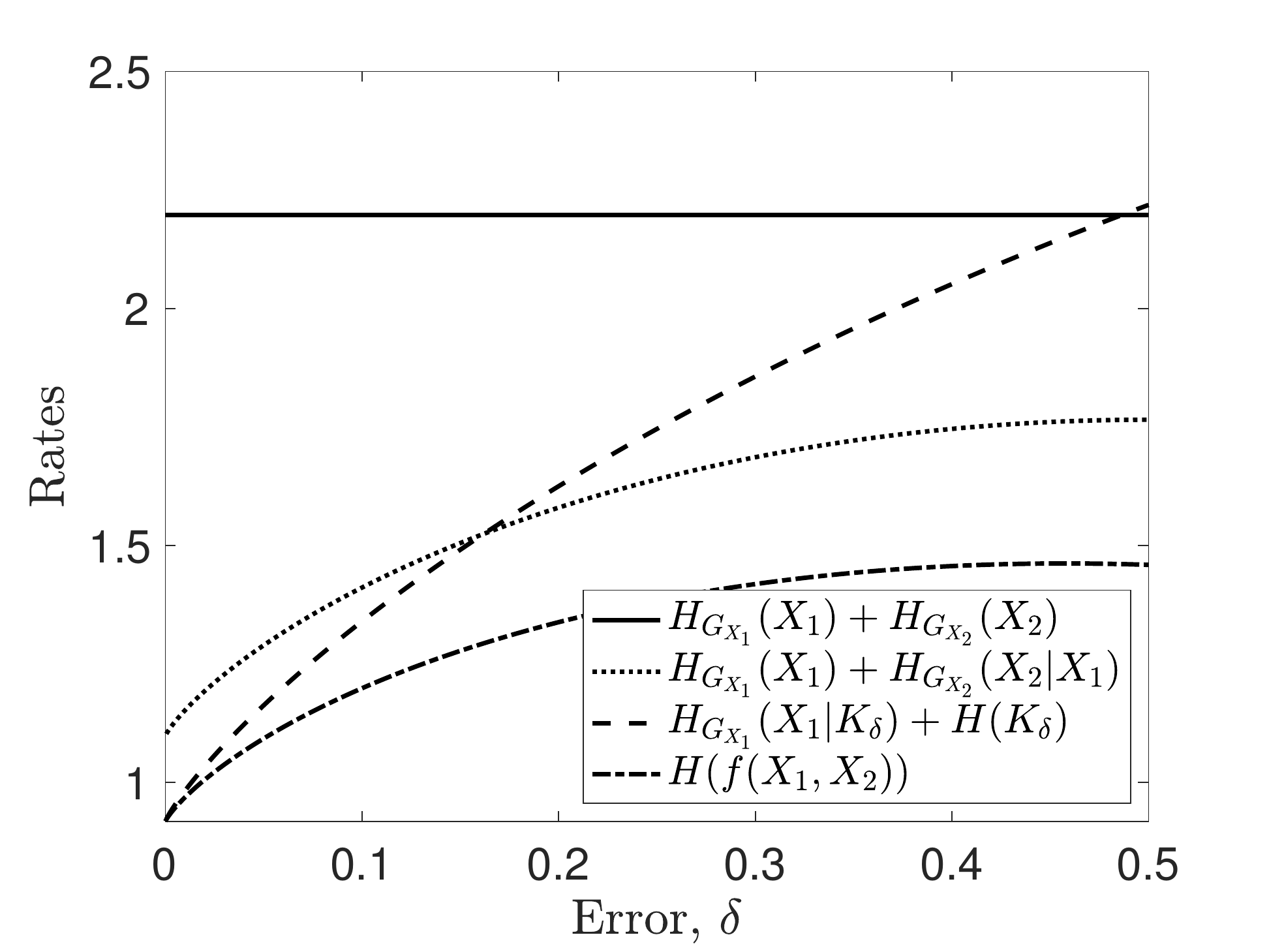}
    \includegraphics[width=0.49\textwidth]{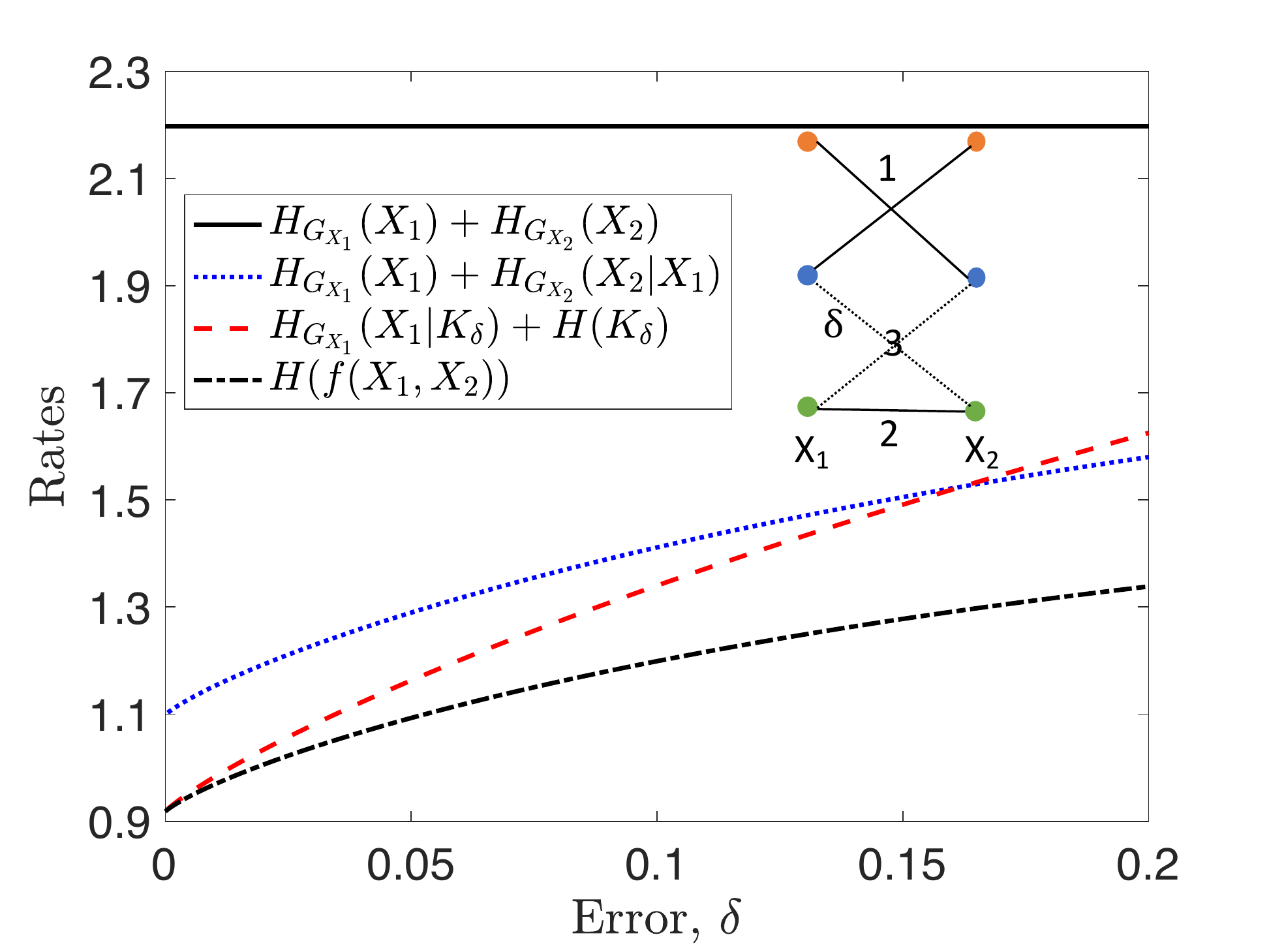}
    \caption{Bipartite symmetric graph $G_f$ in Example \ref{SymmetricLowError} with a cross-over probability $2\delta$ between two bipartitions. The right figure is the zoomed-in graph where $\delta\in(0,\,0.2)$.}
    \label{fig:BipartiteGraph_LowDelta}
\end{figure*}

\begin{ex}\label{SymmetricHighError}Consider the probability matrix $P_{X_1,\, X_2}$ and the table of permutation invariant function outcomes $F$ given as follows:
\begin{align}
P = \frac{1}{3} \begin{bmatrix} 0 & 1-\delta & \delta \\ 1-\delta & 0 & \delta \\ \delta & \delta & 1-2\delta\end{bmatrix},\quad F= \begin{bmatrix} X & 1 & 4 \\ 1 & X & 3\\ 4 & 3 & 2\end{bmatrix},  \nonumber 
\end{align}
where $\delta\in (0,\,1)$. The bipartite graph $G_f$ is shown in Fig. \ref{BipartiteGraphs_HighDelta} where the vertices are listed in the presented order.

\begin{figure}[t!]
    \centering
    \includegraphics[width=0.25\columnwidth]{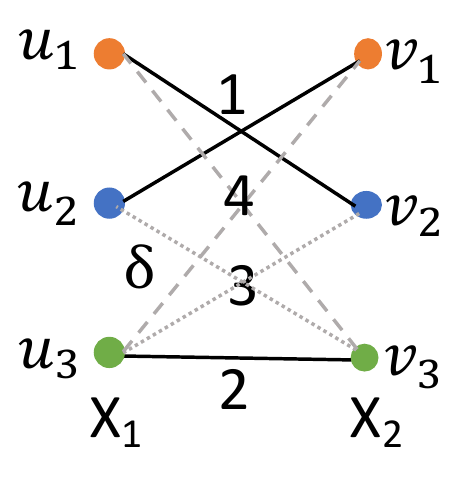}
    \caption{Bipartite graph $G_f$ for Example \ref{SymmetricHighError} where each equivalence class (node) in either characteristic graph is denoted by a different color, each dotted or dashed edge has a probability $\delta$ and the function outcomes are indicated on the edges.}\label{BipartiteGraphs_HighDelta}
\end{figure}

The entropy of the function for given $P_{X_1,\, X_2}$ and $F$ is $H(f(X_1,\, X_2)) = h\left(\frac{2-2\delta}{3},\, \frac{1-2\delta}{3},\, \frac{2\delta}{3},\, \frac{2\delta}{3}\right)$. A trivial rate upper bound is $$H_{G_{X_1}}(X_1)+H_{G_{X_2}}(X_2) = 2\log 3$$ as in Example \ref{SymmetricLowError}. The sum rate required from the graph coloring approach to compute the function in the case of no helper is 
\begin{align}
\mathcal{R}_{f(X_1,\, X_2)}^{\chi}=H_{G_{X_1}}(X_1)+H_{G_{X_2}}(X_2\vert X_1) = \log(3)+\Big(\frac{1}{3}\cdot h(\delta)+\frac{1}{3}\cdot h(\delta)+\frac{1}{3}\cdot h(\delta,\, \delta,\, 1-2\delta)\Big).\nonumber    
\end{align}
The helper-based scheme which distinguishes the bipartitions requires an asymptotic rate of $H(K_{\delta})=2h\left(\frac{\delta}{3},\, \frac{\delta}{3},\, 1-\frac{2\delta}{3}\right)$ bits per use. Hence, the sum rate required to compute the function with a helper that exploits the common information is given as
\begin{align}
\mathcal{R}_{f(X_1,\, X_2)}^H=H(K_{\delta})+H_{G_{X_1}}(X_1\vert K_{\delta}) = 2h\left(\frac{\delta}{3},\, \frac{\delta}{3},\, 1-\frac{2\delta}{3}\right)+h\left(\frac{2}{3}\right).\nonumber
\end{align}
\end{ex}
We compare the rates of entropy coding, coloring, and helper-based schemes in Fig. \ref{fig:BipartiteGraph_HighDelta}. 
The benefit of helper is eminent when $\delta$ is low. The helper needs to identify whether the edge between the bipartitions exists as function of $\delta$, which requires $H(K_{\delta})=2h\left(\frac{\delta}{3},\, \frac{\delta}{3},\, 1-\frac{2\delta}{3}\right)$ bits per use asymptotically. Note that compared to Example \ref{SymmetricLowError} in Fig. \ref{fig:BipartiteGraph_LowDelta}, 
the rate required from the helper to identify the right bipartition is higher because in this example the graph has a higher likelihood to be complete and it is preferred that the sources jointly compress their characteristic graphs. As $\delta$ increases, the graph becomes 
connected with high probability, and the gap between the rates of the helper-based scheme and entropy coding increases.

\begin{figure*}[t!]
    \centering
    \includegraphics[width=0.49\textwidth]{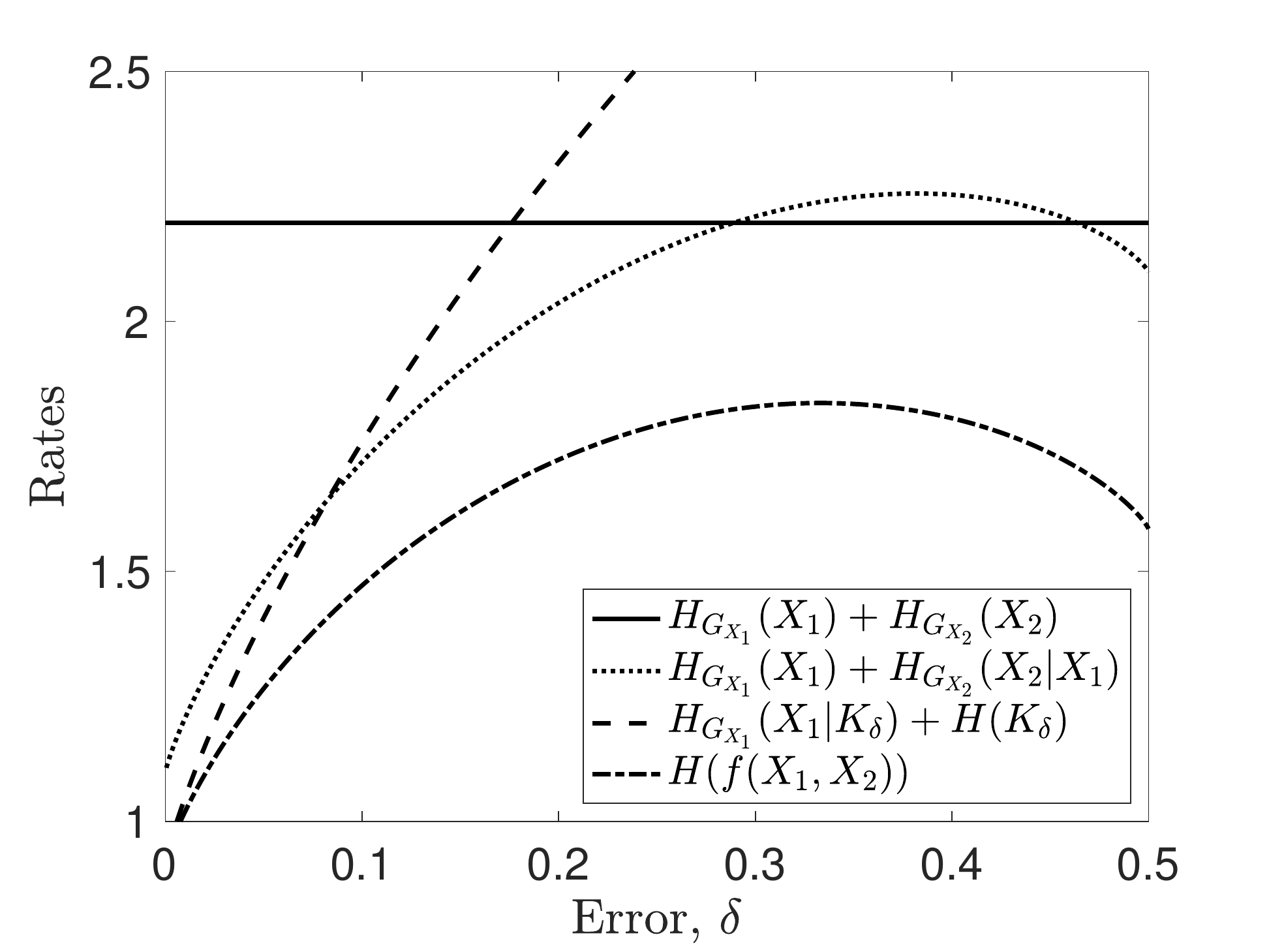}
    \includegraphics[width=0.49\textwidth]{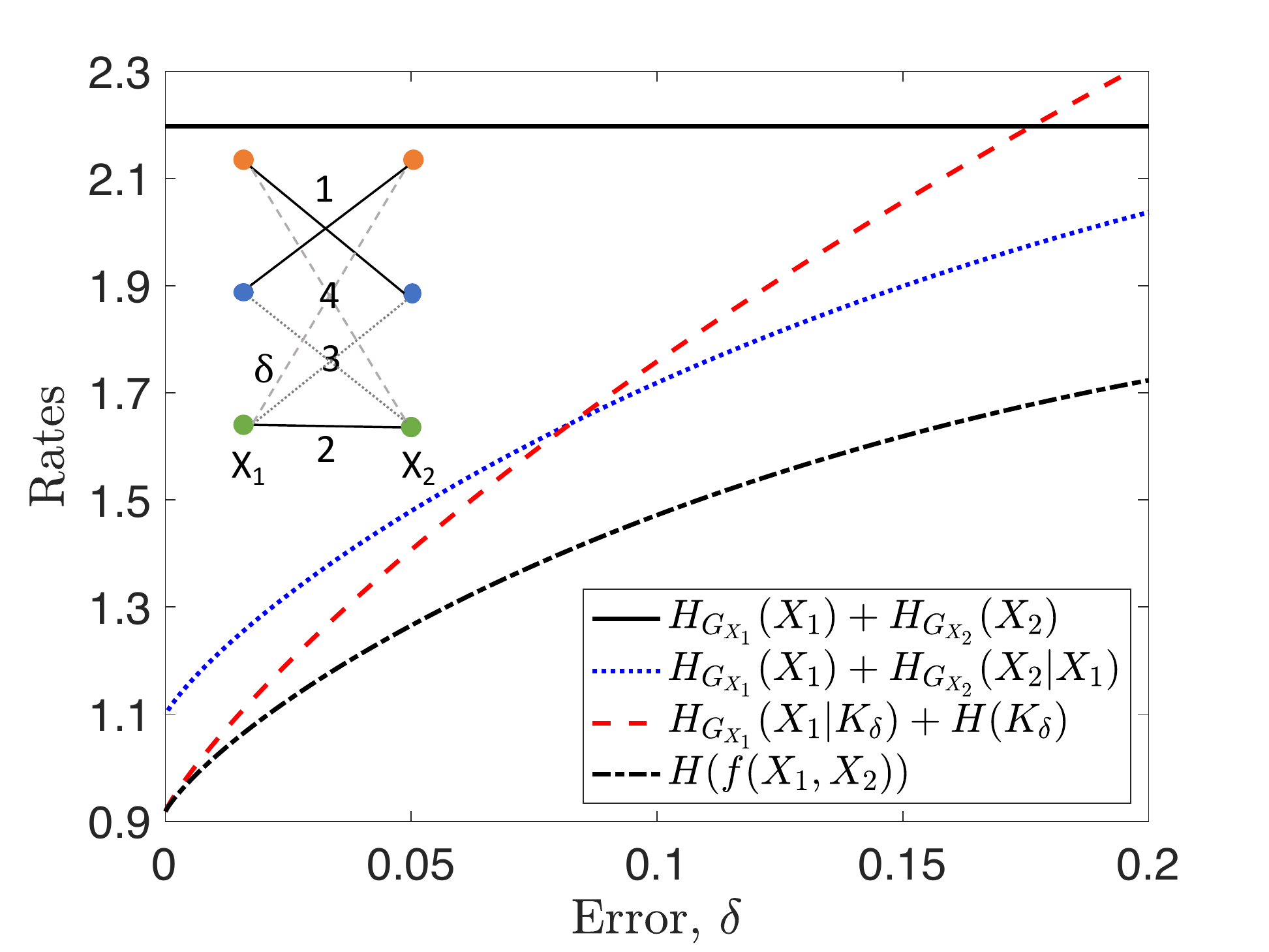}
    \caption{Bipartite symmetric graph $G_f$ in Example \ref{SymmetricHighError} with an aggregate cross-over probability $4\delta$ between two bipartitions. Note that each node in $G_f$ has a higher degree than the nodes in Example \ref{SymmetricLowError}. The right figure is the zoomed-in graph where $\delta\in(0,\,0.2)$.}
    \label{fig:BipartiteGraph_HighDelta}
\end{figure*}

\subsection{Exploiting the Function for Effective Identification of Correlated Structures}
\label{sect:helper_identification}

In this part, we consider another interpretation in which a helper can leverage the structure imposed by the function $f(X_1, X_2)$. To capture the helper or the common information variable that can distinguish the substructures in $G_f$, we use the notation $K_S$. It is easy to note that $K_S$ is different from the set of metrics $K_{\delta}$, $K_B$ and $K_{X_1,\, X_2}$ defined earlier in the paper.

We next focus on distributed computing of $f(X_1,\,X_2)=|X_1|+|X_2|$ with a joint alphabet $\mathcal{X}=\{-2,\,-1,\,0,\,1,\,2\}$. Given $|\mathcal{X}|=3$, the bipartite graph $G_f$ has $3$ nodes in each partition and at most $9$ edges. There are $6$ distinct function outcomes possible due to the symmetry induced by permutation invariance. However, $F$ is not symmetric as the marginal distributions are distinct and the joint source distribution is not necessarily symmetric, i.e.,  $P_{X_1,\, X_2}(X_1=x_1,\,X_2=x_2)\neq P_{X_1,\, X_2}(X_1=x_2,\,X_2=x_1)$. 
To capture the correlation of the source characteristic graphs $G_{X_1}$ and $G_{X_2}$, we next detail two constructions where it suffices if only one source encodes in the presence of a helper that distinguishes a preferred subset of edges, which ameliorates the functional compression of the source variables. 

\begin{ex}\label{Petersen}{\bf Petersen graph.}
Consider the bipartite graph $G_f$ shown in Fig. \ref{fig:Petersen} (left), which is represented by a pentagram within a pentagon, and has five crossings that represent the entries such that $P_{X_1,\, X_2}(X_1\in u ,\, X_2\in v) >0$. The vertices of the pentagon and the pentagram correspond to the characteristic graphs of $G_{X_1}$ and $G_{X_2}$, respectively. The Petersen graph has five edges between $G_{X_1}$ and $G_{X_2}$ of $G_f$. Next to this configuration, we provide an equivalent representation where the function outcomes are indicated on the edges which connect the three equivalence classes $u_i\in U$ and $v_i\in V$ of $G_{X_1}$ and $G_{X_2}$, respectively. The joint distributions of the colors $P_{X_1,\, X_2}$, with entries ordered with respect to vertex colors (orange, blue, and green), and the function outcomes $F$ are given as 
\begin{align}
P_{X_1,\, X_2} = \begin{bmatrix} 0 & p\frac{1-2p}{1-p} & \frac{p^2}{1-p} \\ \frac{p}{2} & 0 & \frac{p}{2} \\ 1-2p & 0 & 0 \end{bmatrix},\quad F=\begin{bmatrix} X & 1 & 3 \\ 1 & X & 2\\ 3 & X & X\end{bmatrix}. \nonumber
\end{align}

\begin{figure*}[t!]
    \centering
    \includegraphics[width=0.95\textwidth]{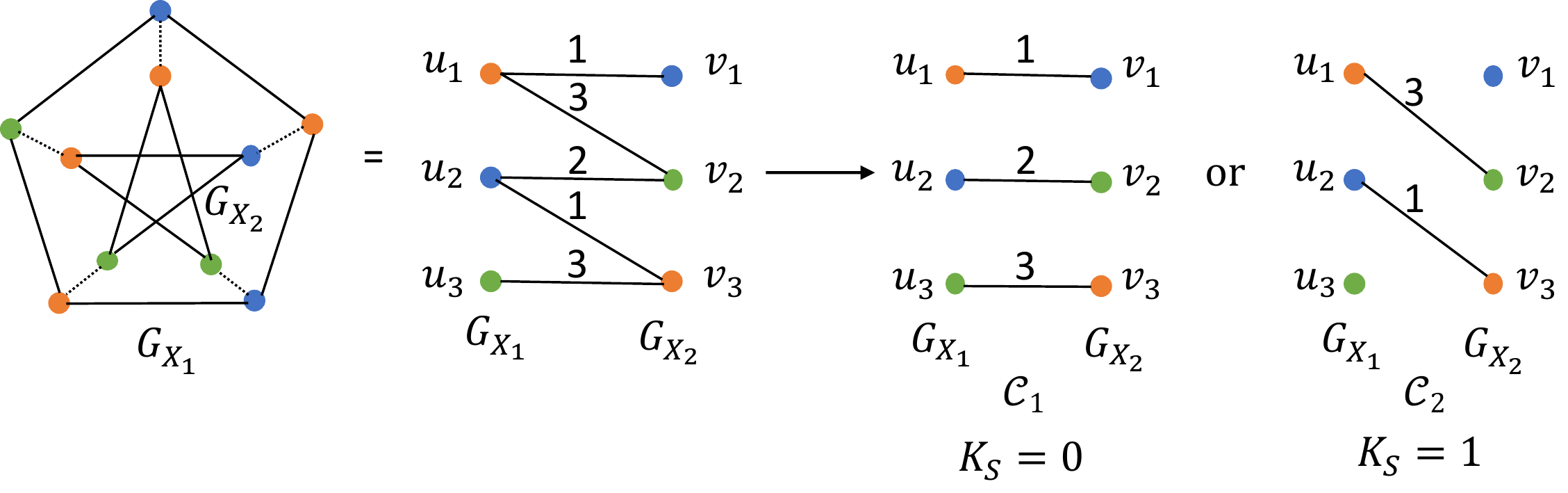}
    \caption{(Left) Petersen graph and (Right) its decomposition.}
    \label{fig:Petersen}
\end{figure*}

From the joint distribution table, $P_{X_1}=(p,\, p,\, 1-2p)$ and $P_{X_2}=(1-\frac{3p}{2},\, p\frac{1-2p}{1-p},\, \frac{p^2}{1-p}+\frac{p}{2})$, respectively for $X_1$ and $X_2$.

In this example, $f(X_1,\, X_2)\in\{1,\,2,\,3\}$ where the probabilities at the respective coordinates are $P_{f(X_1,\, X_2)}=\Big(\frac{p}{2}+p\frac{1-2p}{1-p},\,\frac{p}{2},\,1-2p+\frac{p^2}{1-p}\Big)$. Hence, the entropy of the function from the given joint distribution $P_{X_1,\, X_2}$ and the outcomes table $F$ is $$\mathcal{R}_{f(X_1,\, X_2)}=H(f(X_1,\, X_2))=h\Big(\frac{p}{2}+p\frac{1-2p}{1-p},\,\frac{p}{2},\,1-2p+\frac{p^2}{1-p}\Big).$$

The average number of colors needed to represent the sum function from each source is $H_{G_{X_1}}(X_1)=H(X_1)=h(p,p,1-2p)$ and $H_{G_{X_2}}(X_2)=H(X_2)=h\Big(\frac{p}{2}+1-2p,p\frac{1-2p}{1-p},\,\frac{p^2}{1-p}+\frac{p}{2}\Big)$. This is because each maximal independent set is a singleton and $X_2$ and the maximal independent set $W_1$ containing $X_1$ determine $f(X_1,\, X_2)$, and similarly $X_1$ and the maximal independent set $W_2$ containing $X_2$ determine $f(X_1,\, X_2)$. Hence, the following rate is achievable
\begin{align}
\label{cond_entropy_petersen_given_X2}
H_{G_{X_1}}(X_1\vert X_2)+H_{G_{X_2}}(X_2)&=
\left(1-\frac{3p}{2}\right)\cdot h\Big(\frac{p}{2-3p}\Big)+\Big(p\frac{1-2p}{1-p}\Big)\cdot 0\nonumber\\
&+\Big(\frac{p^2}{1-p}+\frac{p}{2}\Big) \cdot h\Big(\frac{1-p}{1+p}\Big)+H(X_2).
\end{align}
Similarly, the following rate is also achievable
\begin{align}
\label{cond_entropy_petersen_given_X1}
H_{G_{X_2}}(X_2\vert X_1)+H_{G_{X_1}}(X_1)=p\cdot h\Big(\frac{p}{1-p}\Big)+ p\cdot 1+ (1-2p)\cdot 0+H(X_1). 
\end{align} 
Note that the rates in (\ref{cond_entropy_petersen_given_X2}) and (\ref{cond_entropy_petersen_given_X1}) are identical because $H_{G_{X_2}}(X_2\vert X_1)=H(X_2\vert X_1)$ and $H_{G_{X_1}}(X_1\vert X_2)=H(X_1\vert X_2)$, i.e., the chain rule holds for this example. Hence, the total rate of the graph coloring approach without any helper is
\begin{align}
\mathcal{R}_{f(X_1,\, X_2)}^{\chi}=H_{G_{X_1}}(X_1\vert X_2)+H_{G_{X_2}}(X_2).\nonumber
\end{align}

In case of a helper-based scheme exploiting the structure imposed by the function, we assume that the helper distinguishes the edges $(u_1,\,v_2),(u_2,\,v_3)\in E$ that produce function outcomes $f(u_1,\,v_2)=3$ and $f(u_2,\,v_3)=1$ with probabilities $\frac{p^2}{1-p}$ and $\frac{p}{2}$, respectively as given in $P_{X_1,\, X_2}$, from the remaining edges. In other words, the helper decomposes $G_f$ into $\mathcal{C}_1$ and $\mathcal{C}_2$ as shown in Fig. \ref{fig:Petersen}. These graphs have probabilities $P(\mathcal{C}_1)=p\frac{1-2p}{1-p}+\frac{p}{2}+1-2p$, and $P(\mathcal{C}_2)=\frac{p^2}{1-p}+\frac{p}{2}$, which are captured by the indices $K_S=0$ and $K_S=1$, respectively. Hence, the rate required from the helper leveraging the structure imposed by $f(X_1,\, X_2)$ is 
\begin{align}
\label{helper_rate_petersen}
H(K_S)=h\Big(\frac{p^2}{1-p}+\frac{p}{2}\Big).    
\end{align}

In this setup, given the helper with rate $H(K_S)$ it is sufficient when only one source encodes to distinguish the function outcome, which can be seen in Fig. \ref{fig:Petersen}. If source $X_1$ encodes, the rates required to compress $\mathcal{C}_1$ and $\mathcal{C}_2$ are given as 
\begin{align}
\label{graph_C_entropies_petersen_source_1}
H(\mathcal{C}_1)=H_{G_{X_1}}(X_1),\quad \mbox{and}\quad H(\mathcal{C}_2)=h\left(\frac{P_{X_1}(1)}{P_{X_1}(1)+P_{X_1}(2)}\right)=h\Big(\frac{1}{2}\Big)=1.   
\end{align}
If instead of source $X_1$, source $X_2$ encodes, the rates required to compress $\mathcal{C}_1$ and $\mathcal{C}_2$ are 
\begin{align}
\label{graph_C_entropies_petersen_source_2}
H(\mathcal{C}_1)=H_{G_{X_2}}(X_2),\quad \mbox{and}\quad H(\mathcal{C}_2)=h\left(\frac{P_{X_2}(1)}{P_{X_2}(1)+P_{X_2}(3)}\right)=h\Big(\frac{1-\frac{3p}{2}}{1-p+\frac{p^2}{1-p}}\Big).   
\end{align}

\begin{figure*}[t!]
\centering  
\includegraphics[width=0.5\textwidth]{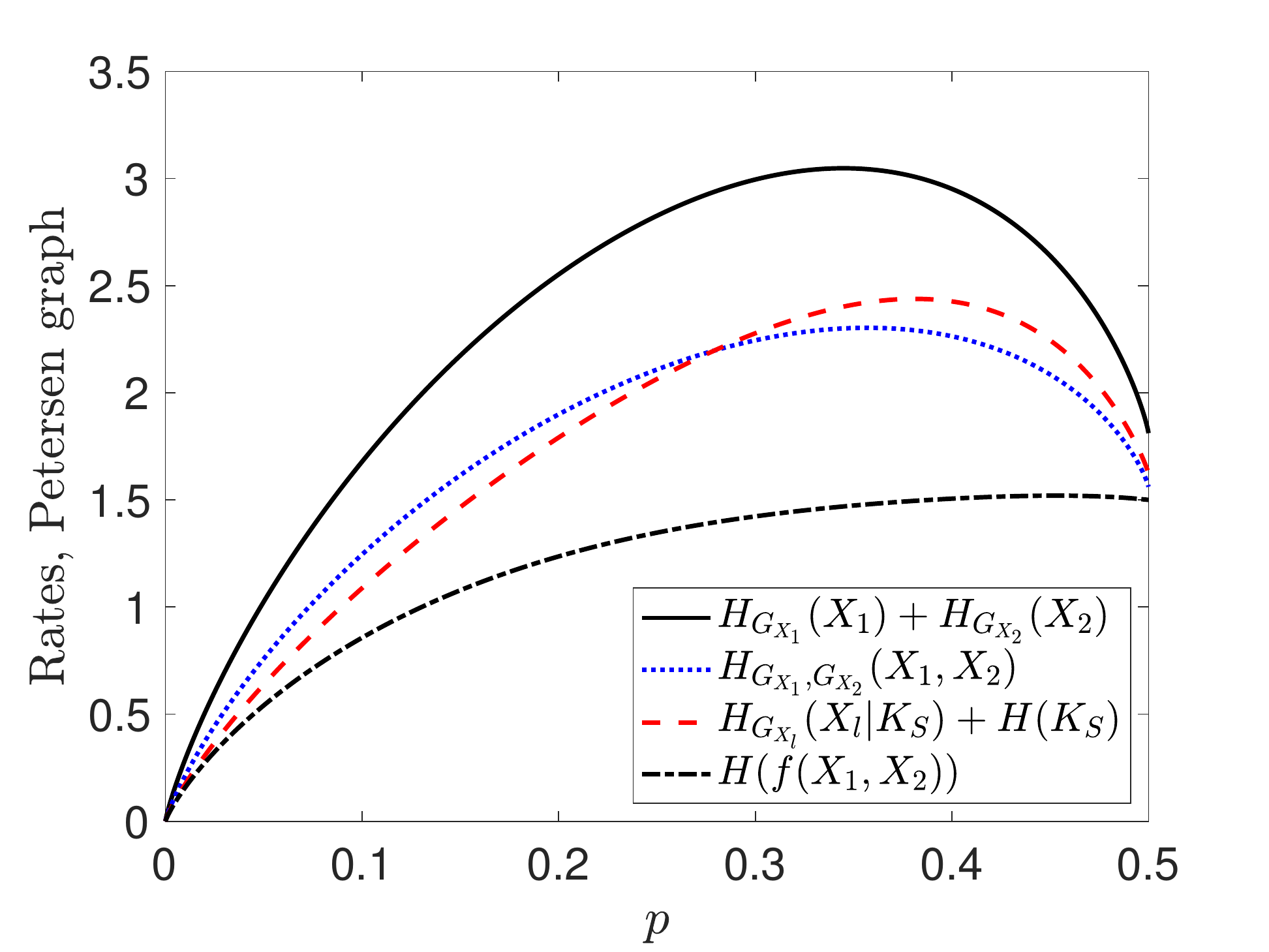}
\caption{Rates for Petersen graph versus $p$.}
\label{fig:petersen_sim}
\end{figure*}

In this setting, in the distributed implementation, using (\ref{helper_rate_petersen}), (\ref{graph_C_entropies_petersen_source_1}), and (\ref{graph_C_entropies_petersen_source_2}) the sum rate of the helper-based model for recovering $f(X_1,\, X_2)$ is
\begin{align}
\label{Rate_Petersen}
\mathcal{R}_{f(X_1,\, X_2)}^H=\min[H_{G_{X_1}}(X_1),\, H_{G_{X_2}}(X_2)]+H(K_S),
\end{align}
that is, one of the sources sends $\mathcal{C}_1$, and the helper sends $\mathcal{C}_2$ as correction. Here, $K_S$ denotes the index that indicates whether a correction is needed and $H(K_S)$ is the rate needed for correction.

Note that the helper-based approach of Sect. \ref{sect:helper_bipartitions} that encodes the bipartitions in $G_f$ using variable $K_B$ is different from the approach of encoding $\mathcal{C}_1$ and $\mathcal{C}_2$ as presented here. In the former, once a bipartition is given a source needs to encode for computing the function conditional on the index $K_B$. In the latter, several (low probability) edges in $G_f$ are removed before compression to ensure functional compression using either one of the sources, and those removed edges are encoded separately by the helper.

Joint functional compression provides significant savings as sources can exploit the dependence. However, there is no clear advantage of using a helper over joint compression because specifying $K$ does not help distinguish the edges with identical outcomes. 
If the helper instead distinguishes $(u_2,\,v_2)\in E$ from the remaining edges, it requires a rate $H(K_S)=h\left(\frac{p}{2}\right)$ which is less than (\ref{helper_rate_petersen}) for small $p$. Despite the lowered helper rate, such a scheme incurs a higher sum rate since both sources should transmit. 
\end{ex}

\begin{ex}\label{CorrelatedStar}{\bf Correlated star graph.}
Consider the bipartite graph $G_f$ shown in Fig. \ref{fig:correlated_star}, where the function outcomes are indicated on the edges. Let $P_{X_1,\, X_2}$ and $F$ be
\begin{align}
P_{X_1,\, X_2} = \begin{bmatrix}  p^2 & p^2 & p(1-2p) \\ 0 & \frac{p^2}{1-p} & p\frac{1-2p}{1-p} \\ 0 & 0 & 1-2p\end{bmatrix},\quad F= \begin{bmatrix} 1 & 1 & 3 \\ X & 3 & 2\\ X & X & 2\end{bmatrix},  \nonumber 
\end{align}
where the entries are ordered with respect to vertex colors.  

\begin{figure*}[t!]
    \centering
    \includegraphics[width=0.65\textwidth]{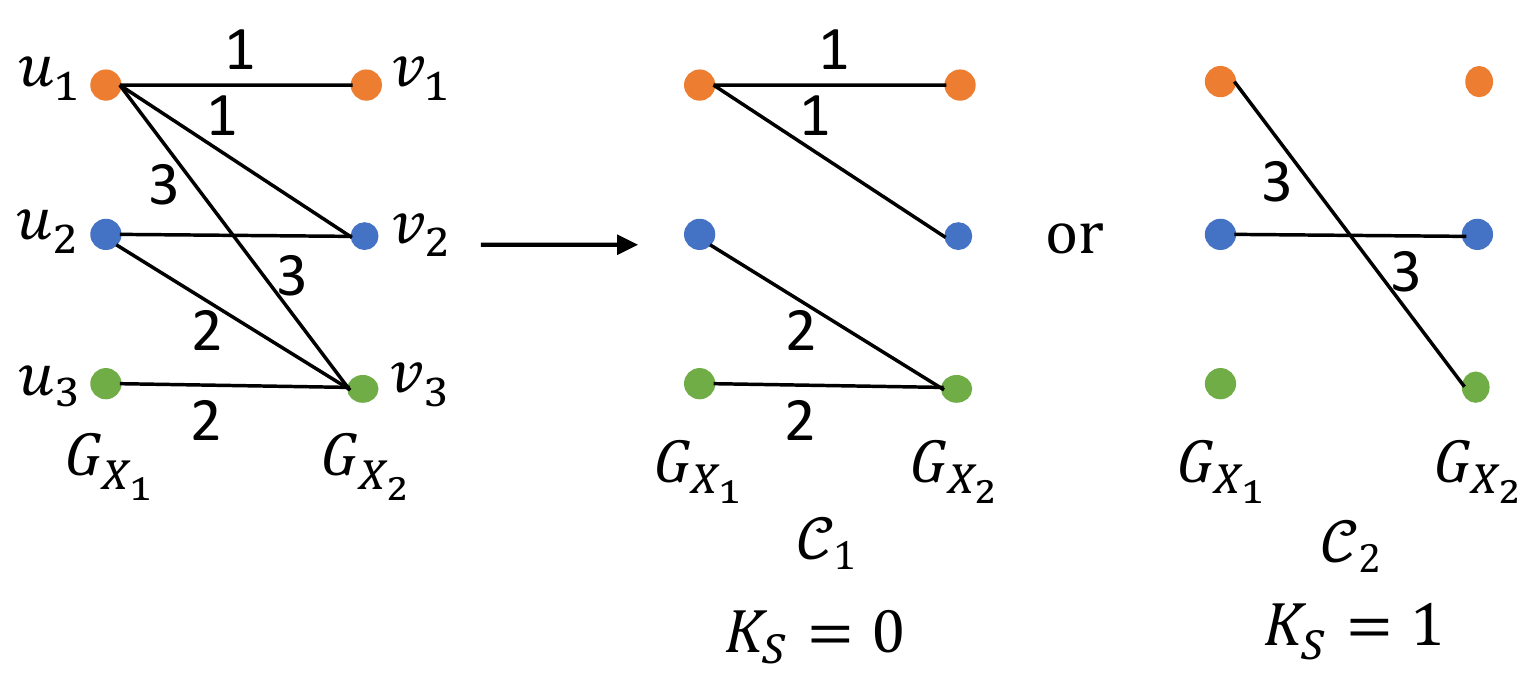}
    \caption{Correlated star graph and its decomposition.}
    \label{fig:correlated_star}
\end{figure*}

From the joint distribution table, the distributions of the colors are $P_{X_1}=(p,p,1-2p)$ and $P_{X_2}=h\Big(p^2,p^2+\frac{p^2}{1-p},p(1-2p)+\frac{1-2p}{1-p}\Big)$ for $X_1$ and $X_2$, respectively. 

In this example, $f(X_1,\, X_2)\in\{1,\,2,\,3\}$ where the probabilities at the respective coordinates are $P_{f(X_1,\, X_2)}=\Big(2p^2,\,\frac{1-2p}{1-p},\,\frac{p^2}{1-p}+p(1-2p)\Big)$. Hence, the entropy of the function from the given joint distribution $P_{X_1,\, X_2}$ and the outcomes table $F$ is
$$\mathcal{R}_{f(X_1,\, X_2)}=H(f(X_1,\, X_2))=h\Big(2p^2,\,\frac{1-2p}{1-p},\,\frac{p^2}{1-p}+p(1-2p)\Big).$$

The average number of colors needed to represent the sum function from each source is $H_{G_{X_1}}(X_1)=h(p,p,1-2p)$ and $H_{G_{X_2}}(X_2)=h\Big(p^2,p^2+\frac{p^2}{1-p},p(1-2p)+\frac{1-2p}{1-p}\Big)$. 

Using $P_{X_1,\, X_2}$ and the table $F$, we can infer the following quantities
\begin{align}
\begin{matrix}  
P(f(X_1,\, X_2)=1\vert X_2\in u_1)=1, & \\ 
P(f(X_1,\, X_2)=1\vert X_2\in u_2)=\frac{1-p}{2-p}, & P(f(X_1,\, X_2)=3\vert X_2\in u_2)=\frac{1}{2-p}, \\ 
P(f(X_1,\, X_2)=2\vert X_2\in u_3)=\frac{1}{1+p-p^2}, & P(f(X_1,\, X_2)=3\vert X_2\in u_3)=\frac{p-p^2}{1+p-p^2},  
\end{matrix}\nonumber
\end{align}
which yields the following achievable rate
\begin{align}
\label{cond_entropy_correlated_star_given_X2}
H_{G_{X_1}}(X_1\vert X_2)+H_{G_{X_2}}(X_2)
&=p^2\cdot 0+\left(p^2+\frac{p^2}{1-p}\right)\cdot h\Big(\frac{1}{2-p}\Big)\nonumber\\
&+\left(p(1-2p)+\frac{1-2p}{1-p}\right)\cdot h\Big(\frac{1}{1+p-p^2}\Big) + H_{G_{X_2}}(X_2).
\end{align}
Similarly, using $P_{X_1,\, X_2}$ and the table $F$, we can infer the following quantities
\begin{align}
\begin{matrix} 
P(f(X_1,\, X_2)=1\vert X_1\in u_1)=2p, & P(f(X_1,\, X_2)=3\vert X_1\in u_1)=1-2p,\\ 
P(f(X_1,\, X_2)=2\vert X_1\in u_2)=\frac{p}{1-p}, & P(f(X_1,\, X_2)=3\vert X_1\in u_2)=\frac{1-2p}{1-p},\\ 
P(f(X_1,\, X_2)=2\vert X_1\in u_3)=1, &  
\end{matrix}\nonumber
\end{align}
which yields the following achievable rate
\begin{align}
\label{cond_entropy_correlated_star_given_X1}
H_{G_{X_2}}(X_2\vert X_1)+H_{G_{X_1}}(X_1)
=p\cdot h(2p) +p\cdot h\Big(\frac{p}{1-p}\Big) +(1-2p)\cdot 0 +H_{G_{X_1}}(X_1).
\end{align}
Note that the sum rates in (\ref{cond_entropy_correlated_star_given_X2}) and (\ref{cond_entropy_correlated_star_given_X1}) are not identical for all $p$ because the chain rule does not hold for graph entropy \cite{OR01}, \cite{AO96}, \cite{Kor73}. Hence, the total rate of joint graph coloring approach without any helper can be given as $$\mathcal{R}_{f(X_1,\, X_2)}^{\chi}=\min[H_{G_{X_1}}(X_1\vert X_2)+H_{G_{X_2}}(X_2),\, H_{G_{X_2}}(X_2\vert X_1)+H_{G_{X_1}}(X_1)].$$

\begin{figure*}[t!]
\centering
\includegraphics[width=0.5\textwidth]{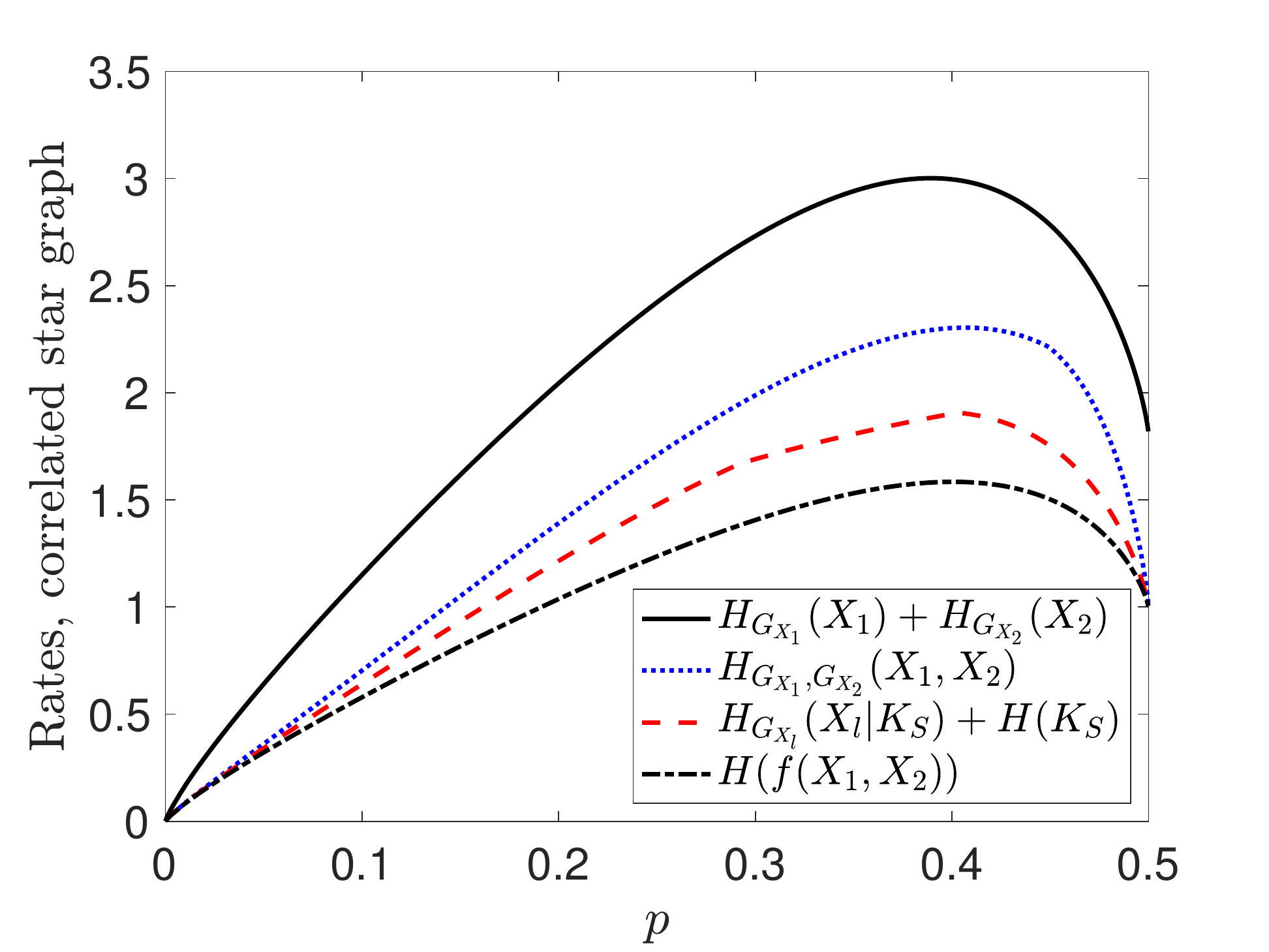}
\caption{Rates for Correlated star graph versus $p$.}
\label{fig:correlated_sim}
\end{figure*}

In case of a helper-based scheme, the helper accounts for the function distribution and decomposes $G_f$ into $\mathcal{C}_1$ and $\mathcal{C}_2$ as shown in Fig. \ref{fig:correlated_star}, where $P(\mathcal{C}_1)=1-\frac{p^2}{1-p}-p(1-2p)$. Hence, the rate required from the helper is
\begin{align}
\label{helper_rate_correlatedstar}
H(K_S)=h\Big(2p^2+\frac{1-2p}{1-p}\Big).
\end{align}
In this setting, only one source needs to transmit as in the Petersen graph. Given $K_S$ we assume that source $X_1$ encodes $\mathcal{C}_1$ and $\mathcal{C}_2$ to distinguish the function outcome. Hence, the rates required to compress $\mathcal{C}_1$ and $\mathcal{C}_2$ are given as
\begin{align}
\label{graph_C_entropies_correlated_star_source_1}
H(\mathcal{C}_1)=h(P_{X_1}(1))=h(p),\quad \mbox{and}\quad H(\mathcal{C}_2)=h\Big(\frac{P_{X_1}(1)}{P_{X_1}(1)+P_{X_1}(2)}\Big)=h\left(\frac{1}{2}\right)=1.
\end{align} 
On the other hand, if instead of source $X_1$, source $X_2$ does the compression, the rates are
\begin{align}
\label{graph_C_entropies_correlated_star_source_2}
H(\mathcal{C}_1)&=h(P_{X_2}(3))=h\Big(p(1-2p)+\frac{1-2p}{1-p}\Big),\quad \mbox{and}\nonumber\\
H(\mathcal{C}_2)&=h\Big(\frac{P_{X_2}(2)}{P_{X_2}(2)+P_{X_2}(3)}\Big)=h\left(\frac{p^2+\frac{p^2}{1-p}}{1-p^2}\right).   
\end{align}

In this setting, using (\ref{helper_rate_correlatedstar}), (\ref{graph_C_entropies_correlated_star_source_1}), and (\ref{graph_C_entropies_correlated_star_source_2}) the total rate of the helper-based model for recovering $f(X_1,\, X_2)$ satisfies
\begin{align}
\label{Rate_CorrelatedStar}
\mathcal{R}_{f(X_1,\, X_2)}^H=\min\Big[h(p),h\Big(p(1-2p)+\frac{1-2p}{1-p}\Big)\Big]+H(K_S),
\end{align}
that is, one of the sources sends $\mathcal{C}_1$, and the helper sends $\mathcal{C}_2$ as correction at a rate $H(K_S)$.

We note that the sum rate in (\ref{Rate_CorrelatedStar}) is relatively small compared to the sum rate in (\ref{Rate_Petersen}) for the Petersen graph because the function distribution ameliorates the computation and the entropies of graphs $\mathcal{C}_1$ and $\mathcal{C}_2$ given in (\ref{graph_C_entropies_correlated_star_source_1}) are small versus the rates in (\ref{graph_C_entropies_petersen_source_1}). Therefore, the gain of the helper scheme is non negligible for the Correlated star graph. 
\end{ex}

Examples \ref{Petersen} and \ref{CorrelatedStar} resemble each other in terms of the ways the graphs $G_f$ are built, as shown by the decompositions in Fig. \ref{fig:Petersen} and Fig. \ref{fig:correlated_star}, respectively. By removing edge $(u_1,v_3)$ from the correlated star graph, we obtain the Petersen graph. However, the $F$ matrices of these graphs are different. 
We illustrate the rates for different encoding schemes as well as the function's entropy versus parameter $p$ in Fig. \ref{fig:petersen_sim} and Fig. \ref{fig:correlated_sim} for the Petersen and the Correlated star graphs, respectively. For the latter the rate required from the helper is close to the fundamental lower bound both at low and high $p$ values.

\section{Conclusions and Directions} 
\label{conclusion}

We proposed a functional common information measure that combines the features of separate extraction (the GKW common information) and joint extraction (through the functional coupling of sources), and relaxes the combinatorial structure of GKW.   
Similar to the multi-source multicast network coding problem, our approach to computing with zero error decoding in general topologies is feasible if and only if the min-cut max-flow theorem is satisfied \cite{ahlswede2000network}. In other words, the encoding rate of sources should not exceed the channel capacity region. While source-channel separation-based functional compression is suboptimal in terms of the encoding rates, modularity ensures low complexity 
and robustness to network failures. The idea can also be extended to functions of multiple sources.

We studied distributed compression for computing permutation invariant functions using a bipartite graph representation. We investigated the achievable rates by incorporating a helper or a common information variable.
We first demonstrated via Examples \ref{SymmetricSum} and \ref{AsymmetricSum} that by dividing the graph into bipartitions exploiting the block-diagonal structure in $P_{X_1,\, X_2}$ via the variable $K_B$, it is possible to effectively encode the function. Note that in some cases, the block-diagonal structure in $P_{X_1,\, X_2}$  is inherent, e.g., in Example \ref{AsymmetricSum}, and thus, the helper-based approach might not provide savings over graph-based compression. 
While we focused in general on identical source alphabets and symmetric $P_{X_1,\, X_2}$, e.g., in Example \ref{SymmetricSum}, higher gains can be attained for disjoint source alphabets because for permutation invariant functions, edges $(u,\, v)$ and $(v,\,u)$ may coexist if and only if $P_{X_1,\, X_2}(X_1\in u ,\, X_2\in v) >0$ and having disjoint alphabets breaks the edge symmetry. This asymmetry ensures a simplified encoding of bipartitions. The role of helper is more pertinent if sources have distinct distributions $P_{X_1}\neq P_{X_2}$ on their characteristic graphs and for functions that are not bijections versus symmetric $P_{X_1,\, X_2}$ with bijective mappings $F$ as in Example \ref{SymmetricSum}. 
We showed via Examples \ref{SymmetricLowError} and \ref{SymmetricHighError} that by employing a helper scheme that exploits the low probability edges (of the bipartite graph) via the variable $K_{\delta}$, high gains in compression are feasible if the cross-over rates across the bipartitions are low. 
For dependent source variables where there is a fewer number of edges, or for surjective functions as in Examples \ref{Petersen} and \ref{CorrelatedStar}, via the variable $K_S$ the helper scheme explores the structure brought by the function, and it ensures higher rate savings in computing, via sending refinements. 

In functional compression, modularized schemes that separate source and channel coding render suboptimal performances in more general network topologies, e.g., non-tree networks or tree topologies, where the source graphs are dependent. 
Our bipartite graph compression technique might be preferred since building source graphs do not necessarily rely on separation. 
However, the implementation complexity grows with the network size and the combinatorial structure, which may sacrifice the performance. 
If the cross-over rates across bipartitions are small, the advantage of our helper-based approach is that it eliminates the need for joint decoding with high complexity, without sacrificing the performance. 
We aim to generalize the canonical model with two characteristic graphs to multiple sources by constructing multipartite graphs, and 
to more general topologies. For instance, if a permutation invariant function $\Phi$ is in the form $\Phi(\sum\nolimits_{x\in S} \psi(x))= \Phi_3(\Phi_1(\sum\nolimits_{x\in S_1} \psi(x))+\Phi_2(\sum\nolimits_{x\in S_2} \psi(x)))$, 
it may as well be possible to compute it via intermediate computation. 

Possible extensions include learning the functional common information graphs $\{\mathcal{C}_i\}$ and the set $\mathcal{I}=\{i\}$ in polynomial time, which depends on the function. The helper can to some extent learn the functional common randomness $K_{f(X_1,\, X_2)}$ via feedback. Note that feedback can help provide $K_{f(X_1,\, X_2)}$, and this can also help reduce the extra rates needed to recover the sources $(X_1,\, X_2)$ themselves. Another key challenge is the implementation of the helper model via exploiting the quality of the helper link. In a wireless edge computing scenario, the rate region of a helper is jointly decided by the channel capacity and how much common randomness it can extract, either directly from the sources, e.g., $K_{X_1,\, X_2}$, or via the function. 

K{\"o}rner's graph entropy-based outer bound for computing 
gives the fundamental limits of compression when the cost of computation was insignificant. However, operating on the outer bound might not jointly optimize the tradeoff between communications and computing. In particular, since the achievable coloring-based schemes are based on NP-hard concepts, constructing optimal compression codes imposes a significant computational burden on the encoders and decoders.

While we still lack a unified theory that provides an understanding of the fundamental limits of functional compression, i.e., rate lower bounds or achievable codebook constructions, directions include investigating the role of helper for computing a broader class of functions. A promising aspect is learning the edges in the bipartite graph through feedback from the decoder for efficient compression of graph data for computing.


\section*{Acknowledgment} 
The author gratefully acknowledges the valuable discussions with Prof. Gustavo De Veciana and Dr. Salman Salamatian. The prior work of Salamatian {\em et al.} has inspired the work in the current paper.

\bibliographystyle{IEEEtran}
\bibliography{references}
\IEEEtriggeratref{3}

\end{document}